\documentclass[12pt]{article}
\usepackage{amsfonts, amsmath, amssymb}
\usepackage{amsthm}
\usepackage{graphicx,color}
\usepackage{listings}
\usepackage{algorithm}
\usepackage{algorithmic}
\usepackage{wrapfig}
\usepackage{enumitem}
\usepackage{stmaryrd}
\usepackage{cite}

\newcommand{\lra}{\leftrightarrow}

\newcommand{\TTT}{\mbox{T}}
\newcommand{\FFF}{\mbox{F}}
\newcommand{\XXX}{\mbox{X}}
\newcommand{\MMM}{\mbox{M}}
\newcommand{\KKK}{\mbox{K}_3}
\newcommand{\LLL}{\mbox{L}_3}
\newcommand{\BBB}{\mbox{B}_3}
\newcommand{\bb}{\mbox{B}_2}

\newtheorem{theorem}{Theorem}[section]
\newtheorem{proposition}[theorem]{Proposition}
\newtheorem{lemma}[theorem]{Lemma}

\newtheorem{definition}{Definition}[section] 
\newtheorem{example}{Example}[section]

\title{A Multiple-Valued Logic Approach to the Design and Verification of Hardware Circuits}
\author{Amnon Rosenmann \\
Institute of Mathematical Structure Theory \\
Graz University of Technology, Graz, Austria \\
rosenmann@math.tugraz.at}
\date{}

\begin{document}
\maketitle

\begin{abstract}
%\boldmath
We present a novel approach, which is based on multiple-valued logic (MVL), to the verification and analysis of digital hardware designs, which extends the common ternary or quaternary approaches for simulations.
%In order to operate in the chosen semantics, the binary designs are transformed 
%into MVL designs, in which the boolean type is replaced by the integer type, and 
%the operators AND, OR and NOT are replaced by MIN, MAX and NEG respectively.
%Binary designs can be automatically transformed into MVL designs, and
The simulations which are performed in the more informative MVL setting reveal details which are either invisible or harder to detect through binary or ternary simulations.
In equivalence verification, detecting different behavior under MVL simulations may lead to the discovery of a genuine binary nonequivalence or to a qualitative gap between two designs.
%Two circuits, which are supposed to be binary equivalent, may behave differently under MVL simulations, and analyzing these differences may lead to the discovery of a genuine binary nonequivalence or to a qualitative gap between the designs.
%By performing an MVL simulation, a combinational design becomes a union of trajectories, where each trajectory starts at some input variable and all the nodes along the trajectory are of the same degree of veracity or falsehood.
The value of a variable in a simulation may hold information about its degree of truth and its ``place of birth'' and ``date of birth''.
%In sequential designs one can incorporate temporal data into simulations, such that a value of a variable in a state of the design may contain, in addition to degrees of truth, the ``place'' and ``date of birth'' of this value.
Applications include equivalence verification, initialization, assertions 
generation and verification, partial control on the flow of data by prioritizing and block-oriented simulations.
Much of the paper is devoted to theoretical aspects behind the MVL approach, including the reason for choosing a specific algebra for computations, and the introduction of the verification complexity of a Boolean expression.
Two basic algorithms are presented.
%and discuss several verification task that may benefit from the MVL approach. 
\end{abstract}
%\vspace{10pt}
% no keywords
%\begin{IEEEkeywords}
%\keywords{many-valued logic, hardware verification, boolean functions}
%\end{IEEEkeywords}
%
%

%******************************************************************************************
%******************************************************************************************
%******************************************************************************************

\section{Introduction}
\label{sec:intro}
The verification and analysis of digital hardware (HW) circuits \cite{L05} has long become a major challenge during the design process.
While formal verification methods such as model checking \cite{CGP01} of properties and formal equivalence checking \cite {MM04}, \cite{KSM10} are complete, they can only be applied to designs of limited size.
The traditional and older method of verification through simulations is incomplete, however it can be applied to larger designs.
%, even to the whole chip, but then suffers from the problem of an incomplete and small coverage of the state-space.
%Without being completely accurate, we can characterize the formal methods as exploring the finite but enormously large state-space in a breadth-first manner, whereas simulation traverses the same space in a depth-first manner.
Hybrid verification methods, which combine concrete or symbolic simulations with formal methods, are also common \cite{Bv06}.

In simulations based on ternary logic (see e.g. \cite{DJD07}) the domain of values of each signal is extended to include a  ``don't care'' (sometimes ``unknown'') value $\XXX$.
It is also common to perform simulation based on quaternary logic, which include a fourth ``high-impedance'' $\mbox{Z}$ value.
%, which is ignored in this note as it is not part of the intended logic function.
%(the $\mbox{Z}$ value is also part of some physical three-state gate devices).
Such logics are also used for abstracting symbolic simulations \cite{WD00}, \cite{WDB00}, \cite{Bv06}, in the model checking technique Symbolic Trajectory Evaluation (STE) \cite{SB95}, in the initialization phase and in equivalence verification \cite{RH02}.

The extension to Multiple-Valued-Logic (MVL) beyond $3$ or $4$ values normally refers to representing a collection of bits as a word or a collection of memory elements as a register when performing simulations with hardware description languages  (see e.g. \cite{R96}).
In addition, some memory devices, arithmetic blocks and FPGAs operate with inputs and outputs which are not binary but multiple-valued.
In general, combinational designs which represent Boolean functions of several variables $f : \{0,1\}^n \mapsto \{0,1\}$ are naturally studied for their algebraic or analytic properties as operating on multi-valued domains of words of length $n$ (see, e.g. \cite{O14}).  

The approach presented here is not to use MVL for treating a collection of binary elements as basic units but rather for performing MVL operations on the binary gate-level elements, extending the ternary-based simulations methodology.
%That is, the single binary input, output and memory elements are transformed into ones that are capable of performing MVL operations.
%For example, the output of an AND gate may hold the value $-6$ while its input values are $-6$ and $10$.
The extension is done by adopting the semantics of the standard fuzzy operators (Zadeh operators): the AND, OR and NOT gates are transformed into the minimum, maximum and negation operators, and the binary domain to $\widehat{\mathbb{Z}}$, an extension of the set of integers with $\pm \infty$ (where $0$ is mostly ignored).

%We present here a novel approach which extends the simulation methods that are based on binary, ternary or quaternary logics (see e.g. \cite{DJD07}) with simulation procedures that are based on MVL.
This extension is simultaneously of a refinement and of an abstraction nature.
The refinement comes from the wider domain of values, which can distinguish between designs that are binary equivalent.
In some cases such a distinction refers to differences in the qualities of the designs.
In other cases it can hint to the existence of a binary nonequivalence, which may be difficult to detect.
Since nonequivalence in the MVL setting is easier to find, we can search in the near environment of an MVL nonequivalence for a ``genuine'', i.e. binary, nonequivalence.
We present an algorithm which is based on these ideas.
 
The abstraction side of performing simulations over MVL is due to being able to treat some of the values as both ``care'' and``don't care'', such that the simulation results can be projected both to binary and to ternary logic.
Unlike simulations done in ternary logic, in the more informative MVL the boundary between the ``care'' and ``don't care'' values need not be determined in advance but rather is dynamic and set upon each simulation according to the output value.
This property (as stated in Theorem~\ref{thm:main_general}) is a key factor in applying MVL for the verification of binary designs.    
We would like to emphasize that this kind of fuzziness is not a matter of interpretation.
Once the outcome of a simulation is obtained, the vagueness disappears and the boundary between the ``care'' and the ``don't care'' values is clear. 

Another special characteristic of MVL simulations is that we can incorporate more information into the domain of values, e.g. temporal and space information. 
Thus, whereas in binary logic we can observe the change in values of a specific variable along time, in MVL simulations of sequential designs we can observe also the change in space of a specific value along time.

The picture is the following.
Suppose that the inputs to a combinational design are assigned values which are of distinct absolute values.
Then these absolute values are spread along the design in the form of a spanning forest.
In particular, there is a path leading from each primary output to an input variable.
In sequential designs, the input values may be augmented with ``date of birth'', such that at each state of a simulation sequence the values of the signals represent, in addition to truth degree, the time when these values were first introduced (and we can also know at which input signal).

Applications include equivalence verification, initialization, assertions generation and verification, partial control on the flow of data by prioritizing and blocks-oriented simulations.
Basic algorithms and general directions towards achieving these goals are presented.

A large part of the paper is devoted to the theory behind the MVL approach that we present.
In Section~\ref{sec:MVL} we analyze the type of MVL that meets our needs and its appropriate semantics $\MMM$.
We also discuss the problematics of ternary logic which is commonly used in HW simulations.
In Section~\ref{sec:computation} we prove the fundamental theorem about the information gained from evaluating Boolean expressions over $\MMM$, on which our approach for simulations relies.
These results have strong connection to the Disjunctive Normal Form (DNF) of the Boolean expressions, when the reductions towards DNF are done according to the laws of De Morgan algebras, as demonstrated in Section~\ref{sec:DNF}.
The DNF plays a role in the definition of the verification complexity that we introdcuce in Section~\ref{sec:complexity}.
This kind of complexity refers to the difficulty of functional validation of a Boolean expression, and differs from the usual complexity which relies on the size of the Boolean expression. 

Section~\ref{sec:combinational} deals with performing simulations over $\MMM$ in the verification of combinational circuits.
A basic algorithm for computing maximal abstract valuations is given, and this algorithm can serve within more complex algorithms for different verification tasks.
An example for such an algorithm is one which is devoted to equivalence verification, as described above (searching for binary nonequivalence in the near environment of an $\MMM$-nonequivalence).
In Section~\ref{sec:sequential} we discuss briefly the potential of $\MMM$-based simulations in the verification of sequential circuits, including the importance of including temporal data in the simulation.

\section{A Suitable MVL and its Semantics $\MMM$}
\label{sec:MVL}
Most modern digital computers are based on binary Boolean algebra,  denoted here $\bb$.
It has two values: $\TTT$ (True, 1) and $\FFF$ (False, 0), and operators like $\lnot$ (NOT, negation, complement), $\land$ (AND, conjunction, meet), $\lor$ (OR, disjunction, join).
Other operators may be defined through these operators, e.g. implication $\varphi \rightarrow \psi$ is defined to be $\lnot \varphi \lor \psi$.
%The classical propositional logic is defined over a binary domain consisting of the constants $\TTT$ (True, 1) and $\FFF$ (False, 0).
%Formulas are composed with connectives, or operators in the Boolean algebra $\bb$, from which we mention:

Our goal is to transform circuit designs which are based on $\bb$ to designs which are based on MVL, such that simulations performed on the transformed designs will be more informative than the ones performed on the original designs.
The significant point here is that the information gained through the MVL simulations should be applicable to the original binary designs, since, after all, these are the ones that need to be verified.

First, let us look at the most common extensions, i.e. to ternary logics.
There are several possible such extensions, and we refer here to 3 known ones:
Kleene's ``strong'' logic $\KKK$ \cite{K38}, {\L}ukasiewicz' $\LLL$ \cite{L30} and Bochvar's $\BBB$ \cite{B37} (also known as Kleene's ``weak'' logic).
In addition to $\TTT$ and $\FFF$, they all contain a third value, denoted here by $\XXX$.
The three logics interpret $\XXX$ differently.
\begin{itemize}
\item In $\KKK$ the meaning of $\XXX$ is some ``vague'' value between $\TTT$ and $\FFF$, which is neither $\TTT$ nor $\FFF$.
%(`undefined', in Kleene's original interpretation)
Hence, we have $\XXX \rightarrow \XXX = \lnot \XXX \lor \XXX = \XXX$.
\item In $\LLL$ the value $\XXX$ represents ``uncertainty'': it can be either $\TTT$ or $\FFF$.
Hence, $\XXX \rightarrow \XXX = \TTT$ since $\varphi \rightarrow \varphi$ is a tautology in binary logic.
Note, however, that the two binary equivalent formulas $\varphi \rightarrow \psi$ and $\lnot \varphi \lor \psi$ are not equivalent in $\LLL$: the law of excluded middle does not hold and $\lnot \XXX \lor \XXX = \XXX$.
\item In the logic $\BBB$ $\XXX$ is interpreted as ``meaningless'' (or ``undefined'' in our modern Computer Science terminology).
Hence, any expression that contains at least one $\XXX$ value is evaluated to $\XXX$.
\end{itemize}

A signal in a circuit is supposed to represent some binary value, either $\TTT$ or $\FFF$.
When performing simulations or formal verification over ternary logic, there are two main reasons for assigning the value $\XXX$ to a variable $v$.
\begin{itemize}
\item One is for representing ``uncertainty'', i.e. when the binary value of $v$ is unknown or not supposed to be determined.
\item The other is for expressing ``don't care'', e.g. when the output of an element does not depend on the binary value of $v$, or when we want to abstract away from the concrete setting.
\end{itemize}

Our intention is to extract more information about the binary design when performing MVL simulations, but in a way that conforms with the original (binary) behavior of the system.
%So, which logic and semantics should we adopt if already in the $3$-valued logics there is a variety of seemingly appropriate candidates?
%Well, the extreme interpretation of $\BBB$ is not so helpful when coming to HW verification 
Thus, $\BBB$ is not suited for this purpose because it blocks any extra information that may be learned about the design beyond the fact that there exists some variable with an $\XXX$ value in case the output is $\XXX$.
In $\KKK$ both $v \rightarrow v$ and $\lnot v \lor v$ equal $\XXX$ when $v$ is assigned the value $\XXX$, although the value of the output signal is always $\TTT$ in the circuit itself.
Consequently, $\KKK$ may be less informative (or of higher entropy) than $\bb$ .
Nevertheless, the logic $\KKK$ is the one that prevails in HW verification.    
The same problem with $\lnot v \lor v$ exists in $\LLL$, and in addition, the fact that $\varphi \rightarrow \psi$ and $\lnot \varphi \lor \psi$ are not equivalent in $\LLL$ is another inconsistency with $\bb$.

In order to overcome the limitations of the ternary extensions shown above, we will apply MVL in a (maybe surprising) way that will keep the boundary between the ``don`t care'' and ``care'' values flexible and dynamic.
It will always be possible to map the simulations done in MVL to $\bb$ in a fixed manner and without any vagueness.
On the other hand, each simulation will tell us which values are for sure ``don't care'' for this specific simulation.
The expressions $\varphi \rightarrow \psi$ and $\lnot \varphi \lor \psi$ will be equivalent in the new setting.
Moreover, they will always be evaluated to $\TTT$ when mapped to $\bb$.
When mapped to ternary values with $\varphi$ and $\psi$ mapped to $\XXX$ then $\varphi \rightarrow \psi$ and $\lnot \varphi \lor \psi$ will also be mapped to $\XXX$, as in $\KKK$ (and clearly, if $\varphi$ is mapped to $\FFF$ or $\psi$ to $\TTT$ then $\varphi \rightarrow \psi$ and $\lnot \varphi \lor \psi$ will be mapped to $\TTT$).

%We will see that it is possible to apply MVL in way that corresponds both to $\KKK$ and to $\LLL$ in the following way.
%When mapping from MVL to $\bb$, $a \rightarrow a$ will always be mapped to $\TTT$, as in $\LLL$. 
%However, as in $\bb$ and unlike $\LLL$, $\lnot a \lor a$ will be equivalent to $a \rightarrow a$ in MVL.
%When mapping into ternary logic then it will conform with the logic $\KKK$: if $a$ will be mapped to $\TTT$ or to $\FFF$ then $a \rightarrow a$ will be mapped to $\TTT$; and if $a$ will be mapped to $\XXX$ then $a \rightarrow a$ will also be mapped to $\XXX$.

Now we come to general multiple-valued logics.
These are logics with more than $2$ values, including infinitely-many values \cite{G01}, \cite{B08}.
Such systems were introduced by {\L}ukasiewicz, G$\ddot{\mbox{o}}$del, Post and many others.
Chang \cite{C58}, \cite{C59} introduced MV-algebras, which generalize Boolean algebras, in order to study {\L}ukasiewicz' logics.
Zadeh introduced fuzzy sets and fuzzy logic \cite{Z65}, \cite{Z96}, \cite{NPM99}, \cite{B08}, where the domain of values is infinite: the closed unit interval.
%The interested reader is referred to \cite{G01} which contains an overview of the different multiple-valued logics developed over the years.  

Since we want the MVL simulations to conform with both $\bb$ and $\KKK$, the algebraic laws of these logics should hold in the chosen MVL.
In addition, we need to choose a suitable semantics $\MMM$ for realizing the MVL.
So, first we need two designated elements denoted by $\top$ and $\bot$, corresponding to $\TTT$ and $\FFF$, and three operators $\land$, $\lor$ and $\lnot$.
Then, there should be at least one {\em homomorphism} $p : \MMM \to \bb$ and at least one homomorphism $p : \MMM \to \KKK$,
such that $p(\top) = \TTT$ and $p(\bot) = \FFF$,
(Recall that a homomorphism is a map that respects the operations: $p(a \land b) = p(a) \land p(b)$, $p(a \lor b) = p(a) \lor p(b)$ and $p(\lnot a) = \lnot p(a)$.)

A natural demand is that the following set of laws of De Morgan algebras should hold in $\MMM$:
\begin{enumerate}
\item Commutativity: $a \land b = b \land a$ and $a \lor b = b \lor a$;
\item Associativity: $a \land (b \land c) = (a \land b) \land c$ and $a \lor (b \lor c) = (a \lor b) \lor c$;
\item Idempotence: $a \land a = a$ and $a \lor a = a$;
\item Absorption: $a \land (a \lor b) = a$ and $a \lor (a \land b) = a$;
\item Distributivity: $a \land (b \lor c) = (a \land b) \lor (a \land c)$ and $a \lor (b \land c) = (a \lor b) \land (a \lor c)$;
\item Identity: $a \land \top = a$ and $a \lor \bot = a$;
\item Consumption: $a \land \bot = \bot$ and $a \lor \top = \top$;
\item Duality: $\lnot \bot = \top$ and $\lnot \top = \bot$;
\item Double Negation: $\lnot \lnot a = a$;
\item De Morgan: $\lnot (a \land b) = \lnot a \lor \lnot b$ and $\lnot (a \lor b) = \lnot a \land \lnot b$;
\end{enumerate}
Note that for a minimal set, the first law at each line suffices.
Also, Absorption may be defined through Identity, Distributivity and Consumption.

The question is how to treat the {\em complementation law}: $a \lor \lnot a = \top$ and $a \land \lnot a = \bot$ of Boolean algebras.
It should clearly hold for $a=\top$ and for $s=\bot$.
However, we do not want it to hold for all other values of $\MMM$ (as in MV-algebras, which provide semantics to generalizations of $\LLL$),
%with lattice operations.
because then we will not gain any further information from working over MVL.
%However, since the complementation laws: $a \lor \lnot a = \top$ and $a \land \lnot a = \bot$ of Boolean algebras hold also in MV-algebras, these algebras do not extend $\KKK$ and are not candidates for our chosen logic.
So, we replace the complementation law with the weaker {\em orthocomplementation law}: $a \lor \lnot a = \top$ should hold for $\top$ and $\bot$ but not necessarily for all elements.
It is easy to see that this requirement is satisfied in De Morgan algebras.
%combination of identity, duality and double negation laws.  
%Note, however that by the homomorphism $p : \MMM \to \bb$, $p(a \lor \lnot a) = \TTT$ and $p(a \land \lnot a) = \FFF$.

With the above rules we can form a lattice.
Better though is to have a complete ordered set, so that any two elements of $\MMM$ could be compared, with $\bot$ and $\top$ being the minimal and the maximal elements respectively: $a > \bot$ for every $a \neq \bot$, and $ a < \top$ for every $a \neq \top$.
%It is better that the domain of values of $\MMM$ will be an ordered set (and not just a lattice with partial order), 
%In fact, we would also need a partition of the domain into the values that are mapped to $\TTT$ and values which are mapped to $\FFF$.
Given a lattice, one defines $a \leq b$ if and only if $a \land b = a$ and $a \lor b = b$.
Thus, in an ordered set the operator $\land$ is defined to be the {\em minimum} and $\lor$ is defined to be the {\em maximum}.
By De Morgan law, we have: $a \leq b$ implies $\lnot b \leq \lnot a$, which then implies:
\begin{enumerate}[resume]
\item For all $a, b$: $\; \; a \land \lnot a \leq b \lor \lnot b$ \, .
\end{enumerate}
A system which satisfies the above $11$ laws is called a Kleene algebra (we remark that there exist in the literature other definitions of Kleene algebras).

One possible semantics that meets all the above requirements is that of fuzzy logic, with the set of values being the closed unit interval, with $1$ representing $\top$ and $0$ representing $\bot$, and the operators minimum (for $\land$), maximum (for $\lor$) and complement $a \mapsto 1 - a$ (for $\lnot a$).
Note that another common semantics for fuzzy logic, in which multiplication comes instead of the maximum operation for $\land$, is rejected since when mapping to $\KKK$ it may happen that $a$ and $b$ will be mapped to $\top$ while $a * b$ will be mapped to $\XXX$ -- which is not a homomorphism.
%Thus, we use the simpler operations of Max and Min, and we also notice that the result of these operations equals one of the operands - fulfilling the above requirement.

For convenience, instead of the unit interval of the continuum cardinality we choose for the domain of values of $\MMM$ the countable set $\widehat{\mathbb{Z}} = (\mathbb{Z} \setminus \{0\}) \cup \{ -\infty, \infty \}$, with the operations $\land, \lor$ and $\lnot$ interpreted as minimum, maximum and negation respectively.
The reason for working over $\widehat{\mathbb{Z}}$ instead of over $[0, 1]$, with $0$ instead of $0.5$ as the mid-point of symmetry, is that it enables using the notion of absolute value, which plays a crucial role in the theory that will be presented.
In practice, we do not need the whole range of values of the integers, and a finite symmetric set around $0$ suffices.
The fact that we omit the value $0$ from the domain of values has to do with the above discussion of being able to treat values simultaneously as ``care'' and ``don't care'', and gaining more information from computations.
However, in cases where these considerations do not matter, we may use also the value $0$ when taking into account complexity considerations since $0$ behaves like the value $X$ in ternary logic: it equals its own negation.

%The extreme elements $-\infty$ and $\infty$ satisfy: $-\infty < a$, $a \neq -\infty$, and $\infty > a$, $a \neq \infty$.

In Table~\ref{table:op} we demonstrate the behavior of the operators $\lnot, \land, \lor$ and $\oplus$ (Exclusive-Or, i.e. $a \oplus b := (a \land \lnot b) \lor (\lnot a \land b)$) in $\MMM$.% (these binary operators are symmetric, so only half of the possible combinations of the operands values are given).
%\vspace{2pt}
%\begin{center}
\begin{table}[h!]
\vspace{-15pt}
\centering
\begin{tabular}{ r | r || r | r | r | r | r }
$ \; {\bf a} \; \; $ & \; ${\bf b} \; \;  $ & \; ${\bf \lnot a} \;  $ & \; ${\bf\lnot b} \; $ & \; ${\bf a \land b}$  \;  & \; ${\bf a \lor b}$ \; &  \; ${\bf a \oplus b}$ \\
\hline \hline
 -2 \, & -1 \, & 2  \, & 1 \, & -2 \, \, \, & -1 \, \, \, & -1 \, \, \\
\hline
 -2 \, & 1 \, & 2  \, & -1 \,  & -2 \, \, \, & 1 \, \, \, & 1  \, \, \\
\hline
 -1 \, & 2 \, & 1 \,  & -2 \, & -1 \, \, \, & 2 \, \, \, & 1  \, \, \\
\hline
 1 \, & 2 \, & -1 \, & -2 \,  & 1 \, \, \, & 2 \, \, \, & -1  \, \, \\
 \end{tabular}
\caption{$\MMM$ operators}
 \label{table:op}
\vspace{-15pt}
\end{table}
%
%In addition to interpreting the domain of values with degrees of truth, as is the semantics of fuzzy logic, it is possible to store more information in it (like `birth date' of a value in sequential circuits). %(not to be confused with probabilities, which refer to a different notion).

%We come now to the homomorphisms from $\MMM$ onto $\bb$ and to $\KKK$.
The homomorphism $p: \MMM \to \bb$ is clear: $p(a) = \FFF$ for $a < 0$, $p(a) = \TTT$ for $a > 0$. 
%\TTT & \mbox{\; \; \; for} & a > 0 \, .
%\begin{equation}
%p(a) = \left\{ \begin{array}{rcl}
%\FFF & \mbox{\; \; \; for} & a < 0 \\ 
%\TTT & \mbox{\; \; \; for} & a > 0 \, .
%\end{array}\right.
%\label{eq:MtoB2}
%\end{equation}
%(here, the subset $\{ a \, | \, a < 0 \}$ refers to the negative integers and $\bot$, and $\{ a \, | \, a > 0 \}$ to the positive integers and $\top$).
%It is easy to check that $p$ is indeed a homomorphism with respect to $\land, \lor$ and $\lnot$.
%
%There are countably-many possible homomorphisms onto $\KKK$.
Then, for every $n >0$, $n \in \mathbb{Z}$, we define $p_n : \MMM \to \KKK$ by: 
 \begin{equation}
p_n(a) = \left\{ \begin{array}{rcl}
\FFF & \mbox{\; \; \; for} & a \leq -n \\ 
\XXX & \mbox{\; \; \; for} & -n < a < n \\
\TTT & \mbox{\; \; \; for} & a \geq n \, .
\end{array}\right.
\label{eq:MtoK3}
\end{equation}
%Again, it is easy to check that $p_n$ is indeed a homomorphism.
%Another homomorphism is $p_{\infty}$, sends $-\infty$ to $\FFF$, $\infty$ to $\TTT$, and sends every $a \in \bbbz \setminus \{0\}$ to $\XXX$.
%X in BL is ambiguous: it is related to the uncertainty about the value, but also to don't care: it can be both.
%$a \oplus b$ (exclusive or) is interpreted as $(a \land \lnot b) \lor (\lnot a \land b)$
%

%******************************************************************************************
%******************************************************************************************
%******************************************************************************************

\section{Computation over $\MMM$}
\label{sec:computation}
A {\em valuation} $v$ of the variables $x_1, \ldots, x_n$ in $\MMM$ (that is, in the domain $\widehat{\mathbb{Z}}$) is a mapping $\llbracket x_1 \rrbracket_v = a_1$, $\ldots$, $\llbracket x_n \rrbracket_v = a_n$, $a_i \in \MMM$.
Given an expression (a propositional formula) $\varphi(x_1, \ldots, x_n)$ and a valuation $v$ as above, we denote by $\llbracket \varphi \rrbracket_v$ the evaluation $\llbracket \varphi \rrbracket_v = \varphi(a_1, \ldots, a_n) \in \MMM$ when looking at $\varphi$ as a function $\varphi : \MMM^n \mapsto \MMM$.
The expression $\varphi$ can be represented as a Directed Acyclic Graph (DAG) $G = G_{\varphi}$, and we denote the computation graph of $\llbracket \varphi \rrbracket_v$ by $G(a_1, \ldots, a_n)$.
The leaves of $G(a_1, \ldots, a_n)$ are labeled with $a_1, \ldots, a_n$, its root - with the value $\varphi(a_1, \ldots, a_n)$, and each internal node representing a sub-expression $\psi$ - with the value $\llbracket \psi \rrbracket_v$.
%In what follows we use the extension of the regular notion of {\em absolute value} on $\bbbz$ to $\MMM$ by $| \infty | = | -\infty | =\infty$. 
\begin{proposition}
\label{prop:path}
Let  $\varphi(x_1, \ldots, x_n)$ be an expression and let $\llbracket x_1 \rrbracket_v = a_1, \ldots$, $\llbracket x_n \rrbracket_v = a_n$ be a valuation in $\MMM$.
Let $G(a_1, \ldots, a_n)$ be the corresponding computation graph.
Then, for some $i$, $1 \leq i \leq n$, $| \varphi(a_1, \ldots, a_n)| = | a_i |$, and there exists a path (at least one) from the root of $G$ to a leaf of it, such that the label of each node along this path is of absolute value $| a_i |$. 
\end{proposition}
\begin{proof}
By induction on the composition depth of $\varphi$ and by the fact that the operations of negation, maximum and minimum preserve the absolute value of one of the operands.
%\qed
\end{proof}

To illustrate the result of Proposition~\ref{prop:path}, suppose that the absolute values $| a_1|, \ldots, |a_n|$ are pairwise distinct, and we label all the nodes of the graph $G = G_{\varphi}$ by their values, and the edges - by the label of their initial nodes.
Then we color the vertices of $G$ with $n$ different colors, such that nodes whose labels are of the same absolute value get the same color.
Next, for each non-leaf node $v = u_1 \lor u_2$ or $v = u_1 \land u_2$ ($u_1, u_2$ are the ``input'' nodes of $v$), the value of $v$ equals the value of some input node $u_i$, $i =1$ or $i=2$, (if $v$ equals both inputs then we choose one of them).
Then we color the edge from this input node $u_i$ to $v$ in the same color of $u_i$, and leave the other in-going edge to $v$ uncolored.
If $v = \lnot u$ then we color the edge from $u$ to $v$ in the color of $u$. 
The result is that each subgraph $G_i$, $i = 1, \ldots, n$, consisting of the vertices and edges of the same color, is in the form of a tree, whose root is a primary input (a leaf of $G$).
The union of the disjoint subgraphs $G_i$ forms a {\em spanning forest} of $G$.

When the leaves of $G$ are not of distinct absolute values then still the result is a spanning forest (now we do not have necessarily a specific root for each tree). 
A generalization of this picture of a spanning forest to the case where the operators $\lor$ and $\land$ have more than 2 arguments is straightforward.
\begin{example}
In Fig.~\ref{fig:forest} we can see the combined graph corresponding to the computation of two expressions over $\MMM$.
The operators $\lor$, $\land$ and $\lnot$, interpreted as maximum, minimum and negation in $\MMM$, are represented by the common gate symbols for the same operators.
The result is a computation of a simple combinational circuit design with two outputs.
The additional XOR symbol represents $a \oplus b := (a \land \lnot b) \lor (\lnot a \land b)$.
The valuation of the arguments is of distinct absolute values, and the solid-line colored subgraph forms a spanning forest.
\begin{figure}[hbt]
%\vspace{-25pt}
\centering
\scalebox{0.5}{ 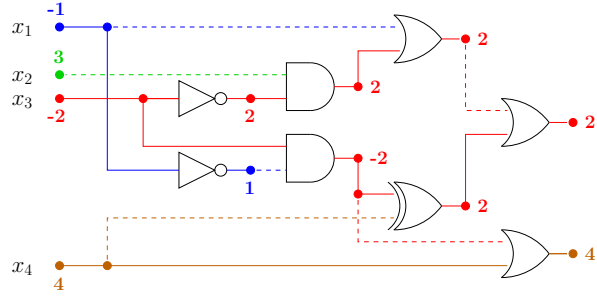 }
\caption{Spanning forest of a computation over $\MMM$}
\label{fig:forest}
\vspace{-15pt}
\end{figure}
\end{example}
%
%The following lemma is straight-forward. 
%

The next theorem shows how more informative are computations done in $\MMM$ compared to those in the binary setting.
It is not only that the range of values is larger.
The qualitative gap is expressed by the fact that the by the result we know for sure about specific arguments that are ``don't care'' when mapped to $\bb$ - they have no influence on the result of that specific computation (there may, however, be more ``don't care'' arguments).
Since the above applies to each sub-expression of the computation, then by examining internal nodes of the computation graph over $\MMM$ we can extract further information about the computation.
\begin{theorem}
\label{thm:main}
Let $| \varphi(a_1, \ldots, a_n)| = | a_i |$ over $\MMM$ and suppose, without loss of generality, that $|a_1| \leq \cdots \leq |a_{i-1}| < |a_i| \leq \cdots \leq |a_n|$.
Then, when evaluated in $\bb$, the value of $\varphi(b_1, \ldots, b_n)$, $b_1, \ldots, b_n \in \{ \TTT, \FFF \}$, does not depend on $b_1, \ldots, b_{i-1}$ as long as for each $j$, $j \geq i$, $b_j = p(a_j)$.
%, where $p$ is the mapping defined in (\ref{eq:MtoB2}).
\end{theorem}
\begin{proof}
If $i=1$ then the claim holds trivially, so let $i > 1$.
Suppose that $\varphi(a_1, \ldots, a_n) = a_i$ (the case where the result is $-a_i$ is similar).
Let $p : \MMM \to \KKK$ be the homomorphism $p = p_{|a_i|}$, hence
%Hence $p(\varphi(a_1, \ldots, a_n)) = p(a_i) \neq \XXX$.
$p(\pm a_1) = \cdots = p(\pm a_{i-1}) = \XXX$.
Therefore, over $\KKK$,
$\varphi(\XXX, \ldots, \XXX, p(a_i), \ldots, p(a_n))$
$= \varphi(p(a_1), \ldots, p(a_{i-1}), p(a_i), \ldots, p(a_n))$
$= p(\varphi(a_1, \ldots, a_n))$
$= p(a_i) \neq \XXX$.
%
%\begin{equation}
%\begin{split}
%p(\varphi(\pm a_1, \ldots, \pm a_{i-1}, a_i, & \ldots, a_n)) \\
%&= \varphi(p(\pm a_1), \ldots, p(\pm a_{i-1}), p(a_i), \ldots, p(a_n)) \\
%\varphi(\XXX, \ldots, \XXX, p(a_i), \ldots, p(a_n))
%&= \varphi(p(a_1), \ldots, p(a_{i-1}), p(a_i), \ldots, p(a_n)) \\
%&= p(\varphi(a_1, \ldots, a_n)) \\
%&= p(a_i) \neq \XXX \, .
%\end{split}
%\end{equation}
As is known, when an expression over $\KKK$ is evaluated to $\TTT$ or to $\FFF$ then the result is invariant to any binary value given to  variables of $\XXX$ values.
%the $\XXX$ values are regarded as `don't care': replacing an $\XXX$ value by either $\TTT$ or by $\FFF$ does not change the outcome of the computation.
%\qed
\end{proof}
In fact, the above theorem follows by Theorem~\ref{thm:main_general} below.
\begin{lemma}
\label{lem:sign}
Let $a, b \in \MMM$. If $|a| > |b|$ then $a > b \Leftrightarrow a > -b$ and similarly $a < b \Leftrightarrow a < -b$.
\end{lemma}
\begin{theorem}
\label{thm:main_general}
Let $| \varphi(a_1, \ldots, a_n)| = | a_i | = b$ over $\MMM$ and suppose, without loss of generality, that $|a_1| \leq \cdots \leq |a_{i-1}| < |a_i| = \cdots = |a_{i+r}| < 
|a_{i+r+1 }| \leq \cdots \leq |a_n|$.
Then $\varphi(a_1, \ldots, a_n)$ is invariant to any change in $a_1, \ldots, a_{i-1}$ (including change of sign) as long as the new values are of absolute value less than $b$.
Neither does any change in value to $a_{i+r+1 }, \ldots, a_n$ affect the result $\varphi(a_1, \ldots, a_n)$, as long as the new values are of absolute value greater than $b$ and there is no change in sign.
%Then, $\varphi(a_1, \ldots, a_n)$ is invariant to any change of sign in $a_1, \ldots, a_{i-1}$: $\varphi(a_1, \ldots, a_{i-1}, a_i, \dots, a_n) = \varphi(\pm a_1, \ldots, \pm a_{i-1}, a_i, \dots, a_n)$.
\end{theorem}
\begin{proof}
%If $i=1$ then the claim holds trivially.
%Suppose that $i > 1$.
We partition the nodes of the computation graph $G(a_1, \ldots, a_n)$ into $3$ groups: (i) those representing operators with operands of absolute value less than $b$; (ii) those with operands of absolute value greater than or equal to $b$; (iii) the nodes representing operators with one operand of absolute value less than $b$ and another operand of absolute value greater than or equal to $b$.

Assume an arbitrary change to the arguments $a_j$, $j < i$, as in the theorem. 
Then a node of the first type may change its value, but will remain of absolute value less than $b$.
A node of the second type will keep its original value.
Finally, by Lemma~\ref{lem:sign}, a node of the third type, representing the maximum or minimum operation, will keep its value if it were of absolute value greater than or equal to $b$, and will stay of absolute value less than $b$ (but perhaps of a different value) if so it were before the change.

Indeed, this is certainly the case for the nodes of level $1$ (from bottom), and, by induction on the height of $G$, the same holds for every node of $G$.  
Since the node representing the result of the computation is labeled with absolute value $b$, it will remain unchanged. 

As for the change of the second type, it is easy to see, again by induction, that it can only affect the nodes of absolute value greater than $b$ (but not the signs), thus being irrelevant to the outcome of the computation. 
%\qed
\end{proof}
%
%

%******************************************************************************************
%******************************************************************************************
%******************************************************************************************

\section{Disjunctive Normal Form over $\MMM$}
\label{sec:DNF}
In Boolean algebra every binary expression $\varphi$ over a set of variables and connectives $\land$, $\lor$ and $\lnot$, can be reduced to an equivalent expression in DNF - a disjunction of conjunctive terms (also called sum-of-products).
Each (conjunctive) term is a conjunction of literals, where a literal is a variable or its negation.
A term is also called an {\em implicant} (or cube) since if $\gamma$ is an implicant of $\varphi$ then for every valuation $v$, $\llbracket \gamma \rrbracket_v = \TTT$ implies $\llbracket \varphi \rrbracket_v = \TTT$.
An implicant $\gamma$ is called a {\em prime implicant} if no subterm of $\gamma$ implies $\varphi$.
The disjunction of all the prime implicants of $\varphi$ is called
%the unique (up to reordering) minimal canonical DNF, and it is called
Blake Canonical Form (BCF), denoted $\mathcal{B}(\varphi)$.
We remark that not all prime implicants are necessarily essential, that is, there may be prime implicants which are covered by other prime implicants, hence $\mathcal{B}(\varphi)$ is not necessarily minimal in number of terms among the DNF that are equivalent to $\varphi$. 
Another canonical DNF is the Full Disjunctive Normal Form (FDNF), denoted $\mathcal{F}(\varphi)$, which consists of all the minterm implicants, that is, each implicant contains all the variables of $\varphi$ (each variable in a complemented or uncomplemented form).

Let us now explore the DNF notion in De Morgan algebras.
%Let us first explore connections between nonequivalence over $\MMM$ and Disjunctive Normal Form (DNF).   
Given an expression $\varphi$ then by the ten rules of De Morgan algebra (see Section~\ref{sec:MVL}) it can be reduced to an equivalent expression $\varphi'$ in DNF.
%That is, for every valuation $v$ in $\MMM$, $\llbracket \varphi \rrbracket_v = \llbracket \varphi' \rrbracket_v$.
The reduction to DNF is, however, more restrictive compared ot the binary case.
By De Morgan rules, subterms of the form $x \land \lnot x$ or $x \lor \lnot x$ cannot be reduced.
% anymore (or, equivalently, the subterm  is irreducible when computing the Conjunctive normal form CNF).
In fact, the only ways by which a conjunctive term can be reduced in size %(after clearing the constants $\top$ and $\bot$)
is by using the idempotence and absorption rules, where the former makes sure that no literal appears twice in a term, and the latter assures that no term is a subterm of another.
This leads us to the following definition. 
\begin{definition}
De Morgan Canonical Form {\upshape (DMCF)} of an expression $\varphi$, denoted $\mathcal{M}(\varphi)$, is the unique (up to reordering), expression which is formed from $\varphi$ by De Morgan reductions and which satisfies:
\begin{itemize}
\item $\varphi$ is in {\upshape DNF};
\item No term of $\varphi$ contains the same literal twice;
\item No term of $\varphi$ is a subterm of another term (in particular, no term appears more than once).
\end{itemize}
\end{definition}
%Each formula $\varphi$ over $\MMM$ can be reduced to a unique (up to reordering) canonical DNF: no term contains the same literal twice and no term is a subterm of another term (in particular, no term appears more than once).
%%
%%%In what follows, when we write about reducing a formula over $\MMM$ we mean that we apply the rules of De Morgan algebra
%
%\begin{definition}
%A formula $\varphi$ is in canonical DNF over $\MMM$ if it is in DNF and cannot be reduced anymore:
%\begin{itemize}
%\item  if $\varphi$ contains the constant $\top$ then $\varphi = \top$, and similarly for $\bot$;
%no term contains the same literal twice and no term is a subterm of another term (in particular, no term appears more than once).
%\end{itemize}
%\end{definition}
%
%Each formula $\varphi$ over $\MMM$ can be reduced to a unique (up to reordering) canonical DNF. 

The reduction to $\mathcal{M}(\varphi)$ is done in a standard way by first driving all negation operators inwards into the literals, then reducing to DNF by the distributive and idemopence rules, and finally deleting terms which contain other terms through the absorption rule (commutativity and associativity are used throughout).

Unlike $\mathcal{B}(\varphi)$, the implicant terms in $\mathcal{M}(\varphi)$ are not necessarily prime implicants.
Another difference is that the terms in $\mathcal{M}(\varphi)$ may be contradictory: containing subterms of the form $x\bar{x}$.
Thus, $\mathcal{M}(\varphi)$ can be expressed as $\mathcal{M}(\varphi) = \mathcal{M}(\varphi)_{imp} \vee \mathcal{M}(\varphi)_{cont}$, where $\mathcal{M}(\varphi)_{imp}$ denotes the disjunction of the implicant terms of $\varphi$, and $\mathcal{M}(\varphi)_{cont}$ denotes the disjunction of the contradictory terms of $\mathcal{M}(\varphi)$.

Let us look at the following examples.
For better readability, in examples we will make use of the following notation: $xy$, $x+y$ and $\bar{x}$ instead of $x \land y$, $x \lor y$ and $\lnot x$ respectively.
The notation $\bar{x}$ for literals will be used also outside of examples. 
%
%For example, $(x \lor y) \land \lnot(\lnot x \land y)$ $\implies$ $(x \lor y) \land (x \lor \lnot y)$   $\implies$ $(x \land (x \lor \lnot y)) \lor (y \land (x \lor \lnot y))$ $\implies$ $x \lor (x \land y) \lor (y \land \lnot y)$ $\implies$ $x \lor (y \land \lnot y)$ (not all reductions were listed).
%Another example is the formula $(x \land \lnot y) \lor y$ that is in canonical DNF over $\MMM$ but in Boolean algebra it can be reduced to $x \lor y$.
\begin{example}
$\varphi = (x +y)\overline{\bar{x}y}$ $\implies$ $(x+y)(x+\bar{y})$ $\implies$ $x(x+\bar{y}) + y(x+\bar{y})$ $\implies$ $xx + x\bar{y} + xy  + y\bar{y}$ $\implies$ $x + y\bar{y} = \mathcal{M}(\varphi)$, which contains the contradictory term $y\bar{y}$.
Of course, over $\bb$ we would have further reduced it to $\mathcal{B}(\varphi) = x$.
\end{example}
\begin{example}
$\varphi = x\bar{y} + y = \mathcal{M}(\varphi)$, with $ x\bar{y}$ not being a prime implicant.
Here, $\mathcal{B}(\varphi) = x + y$.
\end{example}

We write $\varphi \sim \psi$ when $\varphi$ and $\psi$ are $\bb$-equivalent, i.e. $\llbracket \varphi \rrbracket_v = \llbracket \psi \rrbracket_v$ for every binary valuation $v$.
We write $\varphi \sim_{\MMM} \psi$ for $\MMM$-equivalence, i.e$\llbracket \varphi \rrbracket_v = \llbracket \psi \rrbracket_v$ for every valuation $v$ over $\MMM$.
Clearly, $\varphi \sim_{\MMM} \psi$ implies $\varphi \sim \psi$, but not the other way round.
Note that for every valuation $v$ in $\MMM$, $\llbracket \varphi \rrbracket_v = \llbracket \mathcal{M}(\varphi) \rrbracket_v$.

In special cases, again, unlike the case over $\bb$, two different expressions in DMCF may represent equivalent functions over $\MMM$.
For example, the expression $x\bar{x}(y + \bar{y}) + z$ $\implies$ $x\bar{x}y + x\bar{x}\bar{y} + z$ is equivalent to the expression $x\bar{x} + z$.
Similarly, the two $\MMM$-equivalent expressions $x+\bar{x}+y\bar{y}$ and $x+\bar{x}$ are both in DMCF.

The following analysis is done over $\MMM$, but, in fact, for that matter, ternary logic suffices.
Note that over $\MMM$, since disjunction is interpreted as the maximum operator, then an implicant $\gamma$ of $\varphi$ satisfies the following: $0 < \llbracket \gamma \rrbracket_v$ implies $0 < \llbracket \gamma \rrbracket_v \leq \llbracket \varphi \rrbracket_v$, for every valuation $v$.

%The following theorem gives a connection between valuations and canonical DNF over $\MMM$, which demonstrates the qualitative nature of valuations over $\MMM$, a method for distinguishing between HW designs that are $\bb$-equivalent but not $\MMM$-equivalent. 
\begin{lemma}
\label{lemma:implicant}
%Let $\varphi$ and $\psi$ be two formulas which are binary equivalent but not $\MMM$-equivalent, and let $\mathcal{M}(\varphi)$ and $\mathcal{M}(\psi)$ respectively be their canonical DNF over $\MMM$.
Let $\varphi$ and $\psi$ be two expressions,
%satisfying $\varphi \sim \psi$ 
and suppose there is an implicant term $\gamma \in \mathcal{M}(\varphi)_{imp}$ with no subterm of it in $\mathcal{M}(\psi)_{imp}$.
Then there exists a valuation $v$ in $\MMM$ such that $0 < \llbracket \varphi \rrbracket_v > \llbracket \psi \rrbracket_v$.
%> 0$ and hence $\varphi \nsim_{\MMM} \psi$.
%(ii) if $\llbracket \varphi \rrbracket_v > \llbracket \psi \rrbracket_v > 0$ for some valuation $v$ in $\MMM$ then there exists a (binary) prime implicant of $\varphi$ which is a strict subterm of an implicant of $\mathcal{M}(\psi)$.
\end{lemma}
\begin{proof}
%Suppose that $\gamma \in \mathcal{M}(\varphi)_{imp}$ is a subterm of an implicant $\delta \in \mathcal{M}(\psi)_{imp}$.
Let $v$ be the valuation: $\llbracket x_i \rrbracket_v = 2$ for each literal $x_i$ appearing in $\gamma$, $\llbracket x_j \rrbracket_v = -2$ for each literal $\bar{x}_j$ in $\gamma$, and $\llbracket x_k \rrbracket_v = 1$ for each variable $x_k$ not appearing in $\gamma$.
Then $\llbracket \varphi \rrbracket_v = \llbracket \gamma \rrbracket_v = 2$.
% > 1 = \llbracket \psi \rrbracket_v$.
%This is because $\llbracket \varphi \rrbracket_v  = \llbracket \mathcal{M}(\varphi) \rrbracket_v = \llbracket \gamma \rrbracket_v = 2$ since $\llbracket \mathcal{M}(\varphi) \rrbracket_v$ is the maximum of the values of its implicants and it cannot exceed $\llbracket \gamma \rrbracket_v =2$ because $2$ is the maximal absolute value of $v$.
On the other hand, since no term of $\mathcal{M}(\psi)_{imp}$ is a subterm of $\gamma$, each term in $\mathcal{M}(\psi)$, which is positively evaluated by $v$, contains at least one variable which is not in $\gamma$, hence $\llbracket \psi \rrbracket_v  = \llbracket \mathcal{M}(\psi) \rrbracket_v \leq 1$.
% and it equals $1$ because it is positive.
%
%(ii) Let $\llbracket \varphi \rrbracket_v > \llbracket \psi \rrbracket_v > 0$ and let $\gamma \in \mathcal{M}(\varphi)$ such that $\llbracket \varphi \rrbracket_v = \llbracket \gamma \rrbracket_v$.
%Let $\delta$ be a prime implicant of $\varphi$ which is a subterm of $\gamma$, and let $\eta \in \mathcal{M}(\psi)$ such that $\delta$ is a subterm of $\eta$.
%Then $\delta$ is a strict subterm of $\eta$ because otherwise we would have $\llbracket \psi \rrbracket_v \geq \llbracket \eta \rrbracket_v = \llbracket \delta \rrbracket_v \geq \llbracket \gamma \rrbracket_v = \llbracket \varphi \rrbracket_v$, in contradiction to the assumptions.
%\qed
\end{proof}
%
%\begin{example}      
%Let $\varphi = (x \land y) \lor (x \land \lnot y)$ and let $\psi =  (x \land z) \lor (x \land \lnot z)$ be two binary-equivalent formulas in $x, y, z$.
%For the valuation $\llbracket x \rrbracket_v =3$, $\llbracket y \rrbracket_v = 2$, $\llbracket z \rrbracket_v = 1$ we get $\llbracket \varphi \rrbracket_v = 2$ and $\llbracket \psi \rrbracket_v = 1$.
%No term of $\varphi$ is a subterm of a  term of $\psi$.
%However, the BCF of $\varphi$ and $\psi$ is $x$ which is a subterm of e.g. $x \land z$ of $\psi$.
%\end{example}
%
%In the next proposition we show that formulas that are binary equivalent agree on the $\MMM$-evaluations of the implicant part of their canonical DNF whenever the evaluations of the formulas are negative. 

The preceding lemma referred to positive evaluations of expressions.
The next one refers to negative evaluations.
\begin{lemma}
\label{lemma:negval}
%Let $\varphi$ and $\psi$ be a two expressions satisfying $\varphi \sim \psi$.
%, and let $\mathcal{M}(\varphi)$ and $\mathcal{M}(\psi)$ respectively be their canonical DNF over $\MMM$.
If $\varphi \sim \psi$ then $\llbracket \mathcal{M}(\varphi)_{imp} \rrbracket_v = \llbracket \mathcal{M}(\psi)_{imp} \rrbracket_v$ for each valuation $v$ in $\MMM$ for which $\llbracket \varphi \rrbracket_v < 0$ (equivalently, $\llbracket \psi \rrbracket_v < 0$).
\end{lemma}
\begin{proof}
%Let $\mathcal{B}(\varphi)$ and $\mathcal{F}(\varphi)$ be the BCF respectively the FDNF canonical forms (over binary logic) of $\varphi$ (hence, also of $\psi$). 
%
For each implicant $\delta \in \mathcal{M}(\varphi)_{imp}$ there is a prime implicant $\gamma$ of $\mathcal{B}(\varphi)$ (BCF of $\varphi$) which is a subterm of $\delta$. Hence, $\llbracket \delta \rrbracket_v \leq \llbracket \gamma \rrbracket_v$ for each $\MMM$-valuation $v$, and consequently $\llbracket \mathcal{M}(\varphi)_{imp} \rrbracket_v \leq \llbracket \mathcal{B}(\varphi) \rrbracket_v$.
% since in DNF we compute the maximum over the values of the terms.

Similarly, for each implicant $\delta$ in $\mathcal{F}(\varphi)$ (FDNF of $\varphi$) there is an implicant $\gamma \in \mathcal{M}(\varphi)_{imp}$ which is a subterm of $\delta$, hence $ \llbracket \mathcal{F}(\varphi) \rrbracket_v \leq \llbracket \mathcal{M}(\varphi)_{imp} \rrbracket_v$ for each $\MMM$-valuation $v$.

To finish the proof we need to show that when $\llbracket \varphi \rrbracket_v < 0$ then the above inequalities are equalities (note that $\mathcal{B}(\varphi) = \mathcal{B}(\psi)$ and $\mathcal{F}(\varphi) = \mathcal{F}(\psi)$).
Let $\gamma$ be a term in $\mathcal{B}(\varphi)$ such that $\llbracket \gamma \rrbracket_v < 0$.
For each variable $x_k$ that does not appear in $\gamma$, let $l_k = x_k$ if $\llbracket x_k \rrbracket_v > 0$, and $l_k = \bar{x}_k$ if $\llbracket x_k \rrbracket_v < 0$.
By the definition of $\mathcal{F}(\varphi)$, there exists a term $\delta$ in $\mathcal{F}(\varphi)$, such that $\gamma$
is a subterm of $\delta$ and the other literals of $\delta$ are the above $l_k$.
Clearly, since $\llbracket \gamma \rrbracket_v < 0$ and for each of the literals $l_k$, $\llbracket l_k \rrbracket_v > 0$, then $\llbracket \gamma \rrbracket_v = \llbracket \delta \rrbracket_v$.
It follows by the maximum operation in DNF that $\llbracket \mathcal{B}(\varphi) \rrbracket_v \leq \llbracket \mathcal{F}(\varphi) \rrbracket_v$, and by the previous inequality in the opposite direction it is an equality.
%\qed
\end{proof}
\begin{theorem}
\label{thm:BCF}
Let $\varphi$ be an expression satisfying $\mathcal{M}(\varphi) = \mathcal{B}(\varphi)$.
% whose canonical DNF over $\MMM$ is a BCF for binary logic.
Then for any expression $\psi$ satisfying $\varphi \sim \psi$ and for any valuation $v$ in $\MMM$, $| \llbracket \varphi \rrbracket_v | \geq | \llbracket \psi \rrbracket_v |$.
\end{theorem}
\begin{proof}
Let $\mathcal{M}(\varphi)$ and $\mathcal{M}(\psi)$ be DMCF of $\varphi$ and $\psi$ respectively.
If $\llbracket \varphi \rrbracket_v >0$ then since for every implicant $\delta \in \mathcal{M}(\psi)_{imp}$ there exists $\gamma \in \mathcal{M}(\varphi)$ which is a subterm of $\delta$ then $\llbracket \varphi \rrbracket_v \geq \llbracket \psi \rrbracket_v$ as in the proof of Lemma~\ref{lemma:implicant}.

If, on the other hand, $\llbracket \varphi \rrbracket_v < 0$ then by Lemma~\ref{lemma:negval}, $\llbracket \mathcal{M}(\varphi) \rrbracket_v = \llbracket \mathcal{M}(\psi)_{imp} \rrbracket_v$.
Since $\mathcal{M}(\psi)$ may also contain a contradictory part, $\mathcal{M}(\psi)_{cont}$, the inequality follows.
%: $\llbracket \mathcal{M}(\varphi) \rrbracket_v \leq \llbracket \mathcal{M}(\psi) \rrbracket_v <0$.
%\qed
\end{proof}
\begin{example}      
Let $\mathcal{M}(\varphi) = x + y = \mathcal{B}(\varphi) $ and let $\mathcal{M}(\psi) = x\bar{y} + y$ be two $\bb$-equivalent expressions with different DMCF.
Then, for the valuation $\llbracket x \rrbracket_v = 2$, $\llbracket y \rrbracket_v = -1$, we obtain $\llbracket \mathcal{M}(\varphi) \rrbracket_v = 2$ whereas $\llbracket \mathcal{M}(\psi) \rrbracket_v = 1$.
However, when $\llbracket \mathcal{M}(\varphi) \rrbracket_v < 0$ then $\llbracket \mathcal{M}(\varphi) \rrbracket_v = \llbracket \mathcal{M}(\psi) \rrbracket_v$.
For example, when $\llbracket x \rrbracket_v = -2$, $\llbracket y \rrbracket_v = -1$ then $\llbracket \mathcal{M}(\varphi) \rrbracket_v = \llbracket \mathcal{M}(\psi) \rrbracket_v = -1$.
\end{example}

%******************************************************************************************
%******************************************************************************************
%******************************************************************************************

\section{Verification Complexity over $\MMM$}
\label{sec:complexity}
Suppose we want to know the functionality of a Boolean expression by evaluating it on different test vectors over $\MMM$.
The question is how many test vectors are needed in order to transmit to a verifier a complete knowledge of the functionality of the expression, that is, what is the number of test vectors needed for complete functional verification.
This question is related to minimal number of terms in Disjunctive Normal Forms of the expression.
For that purpose, we define three notions of complexity of Boolean expressions: {\em functional complexity}, {\em structural complexity} and {\em verification complexity}
%We define three notions of complexity of Boolean expressions: {\em functional complexity} (or binary complexity, Blake complexity), {\em structural complexity} (or De Morgan complexity) and {\em verification complexity} (or ternary complexity).
(not to be confused with complexity defined as the minimal number of operators, or gates in a circuit representation of the expression, as e.g. in \cite{W87}). 
%(these notions have a specific and narrow meaning, unlike the known notions of Kolmogorov complexity, computational complexity or proof complexity, which are widely used in Computer Science).
%They refer to minimal number of terms in Disjunctive Normal Forms and to minimal number of ternary test vectors needed in order to transmit to a verifier a complete knowledge of the expression.

As before, a Boolean expression is composed of variables and the conjunction, disjunction and negation operators, without constants (even for a tautology or a  contradiction).
%The minimal number of MVL test vectors needed for complete knowledge of an expression is already achieved in ternary logic (that is, when we know in advance what test vectors should be used).
%By the above, it is connected to the canonical DNF.
%So, we define the following notions.
%
In order to gain complete knowledge on the functionality of a Boolean expression $\varphi$ we need to find all the binary vectors $v$ for which $\llbracket \varphi \rrbracket_v = \TTT$ and all the binary vectors $u$ for which $\llbracket \lnot \varphi \rrbracket_u = \TTT$ (equivalently, $\llbracket \varphi \rrbracket_v = \FFF$).
Note that for representing the function only one of the above is needed.
Note also that there is a one-to-one correspondence between the DNFs of $\lnot \varphi$ and the CNFs (Conjunctive Normal Forms) of $\varphi$, so that the number of conjunctive terms in a DNF of $\lnot \varphi$ equals the number of disjunctive terms in the corresponding CNF of $\varphi$.

When testing an expression on binary vectors we need to try all the possible input vectors for complete functional verification.
Over $\MMM$ the number of test vectors that are needed may be much smaller as a consequence of the existence of ``don't care'' variables.
As we saw in the preceding section, there is an inverse relationship between the lengths of the terms in the canonical DNF of an expression and the absolute values of the outcome of the $\MMM$-evaluations of the expression.
Indeed, the shorter the term. the larger is the number of ``don't care'' variables (for that term).
%In order to gain full knowledge of the Boolean functionality of an expression, we can check on which $\MMM$-test vectors its value is positive (equivalently, the corresponding $\bb$-vector satisfies the expression), and on which $\MMM$-valuations its value is negative (i.e. it is evaluated to $\FFF$ in the corresponding binary setting). 

Let us introduce the following notation.
Let $\mathcal{B}_{min}(\varphi)$ be a reduction of $\mathcal{B}(\varphi)$ to a minimal number of prime implicants, which cover all the implicants of $\mathcal{B}(\varphi)$.
Let $\mathcal{M}_{min}(\varphi)$ be a reduction of $\mathcal{M}(\varphi)$ to a minimal number of terms of $\mathcal{M}(\varphi)_{imp}$, which cover all the implicants of $\mathcal{M}(\varphi)_{imp}$ (with $\mathcal{M}_{min}(\varphi)_{cont} = \mathcal{M}(\varphi)_{cont}$).
We denote by  $\#\psi$, for $\psi$ in DNF, the number of (conjunctive) terms it contains.
\begin{definition}
The functional complexity of a Boolean expression $\varphi$ is $$\mathfrak{C}_f(\varphi) := \#\mathcal{B}_{min}(\varphi) + \#\mathcal{B}_{min}(\lnot \varphi)$$.
\end{definition}
\begin{definition}
The structural complexity of a Boolean expression $\varphi$ is $$\mathfrak{C}_s(\varphi) := \#\mathcal{M}_{min}(\varphi) + \#\mathcal{M}_{min}(\lnot \varphi)$$.
\end{definition}
We may say that the functional complexity puts more weight on the semantics of the expression than the structural complexity, while the syntactic part is more emphasized in the structural complexity.
%
%In addition, we define a third kind of complexity. 
\begin{definition}
The verification complexity $\mathfrak{C}_v(\varphi)$ of a Boolean expression $\varphi$ is the number of $\MMM$-valued test vectors needed for complete verification of the binary functionality of $\varphi$.
%the sum of the number of terms in $\mathcal{M}(\varphi)$ and the number of terms in $\mathcal{M}(\lnot \varphi)$.
\end{definition}

Certainly, $\mathfrak{C}_f(\varphi) \leq \mathfrak{C}_s(\varphi)$.
As for the verification complexity, we have the following.
\begin{proposition}
\label{prop:complexity}
$\mathfrak{C}_v(\varphi) \leq \mathfrak{C}_s(\varphi)$.
\end{proposition}
\begin{proof}
Let $\varphi$ be a Boolean expression with $n$ variables.
%Clearly, for each binary test vector $v$, $\llbracket \varphi \rrbracket_v = \FFF$ if and only if $\llbracket \lnot \varphi \rrbracket_v = \TTT$ (and vice versa).
%
%Let $\mathcal{M}(\varphi)$ be the De Morgan canonical DNF of $\varphi$, and let $\mathcal{M}(\varphi)_{imp}$ be its implicant part.
%
It is sufficient to consider the ternary logic $\KKK$ as the MVL over which we form the test vectors.
For each term $\gamma$ of $\mathcal{M}(\varphi)_{imp}$ of length $k$, we form the ternary test vector $v = v(\gamma)$ such that $\llbracket \gamma \rrbracket_v = \TTT$, and such that all the $n-k$ variables that do not appear in $\gamma$ are assigned the value $\XXX$ in $v$.
%We call such a test vector a ternary implicant vector.
This ternary test vector $v$ covers $2^{n-k}$ binary test vectors which satisfy $\varphi$.
Thus, by considering all the implicants $\gamma \in \mathcal{M}(\varphi)_{imp}$, we can find all the binary vectors that satisfy $\varphi$. 
Similarly, we form the ternary test vectors $u = u(\delta)$ for each $\delta \in \mathcal{M}(\lnot \varphi)_{imp}$, such that $\llbracket \lnot \varphi \rrbracket_u = \TTT$, thus finding all the binary vectors that do not satisfy $\varphi$.
%
%The (in general, not disjoint) union of the above subsets of binary vectors covered by the ternary implicant vectors provide all the $2^n$ possible binary vectors.
%Thus, we showed that $\mathfrak{C}_s(\varphi)$ ternary test vectors suffice for complete verification, hence $\mathfrak{C}_v(\varphi) \leq \mathfrak{C}_s(\varphi)$.
%\qed
\end{proof}

We do not know of an example in which $\mathfrak{C}_v(\varphi) < \mathfrak{C}_s(\varphi)$.
Hence, it might very well be that the two notions are identical (that is why we chose the general term ``verification complexity'' without referring to $\MMM$).

The information gained from an $\MMM$-valuation is greater than that of a ternary one.
Suppose that the variables in $\varphi$ are $x_1, \ldots, x_n$ and that a ternary test vector $v$ assigns the value $\XXX$ to $x_1, \ldots, x_k$, while the other variables are assigned binary values.
Suppose also that $\llbracket \varphi \rrbracket_v = \TTT$.
Then we know that there exists an implicant term $\gamma \in \mathcal{M}(\varphi)_{imp}$ of length at most $n-k$, whose variables are among $x_{k+1}, \ldots, x_n$, and whose literals agree with the valuation of $x_{k+1}, \ldots, x_n$.
Over $\MMM$, a possible valuation $w$ is to assign the variables $x_1, \ldots, x_n$ values $a_1, \ldots, a_n$ with increasing absolute values, with some chosen signs to $a_1, \ldots, a_k$ and the signs of $a_{k+1}, \ldots, a_n$ agree with their values in $v$.
We know that $\llbracket \varphi \rrbracket_w =|a_j|$, $j \geq k+1$.
If $j > k+1$ then we know that the following property $P(j)$ holds:
\begin{itemize}
\item $P(j)$: There exists an implicant whose set of variables contains $x_j$, and possibly other variables among $x_{j+1}, \ldots, x_n$, and whose literals agree with the valuation $v$. 
\end{itemize}
Since the upper bound on the length of the implicant is smaller than in the ternary case, then over $\MMM$ we know for sure on more binary vectors which satisfy $\varphi$ (since the variables $x_{k+1}, \ldots, x_{j-1}$ are ``don't care'').
We also know that the following property $N(j+1)$ holds:
\begin{itemize}
\item $N(j+1)$: there is no implicant term of $\varphi$ which contains $x_{j+1}, \ldots, x_n$ and with literals that agree with the valuations $v$.
\end{itemize}
This is implied by the nice property of having a dynamic boundary between the ``care'' values and the ``don't care'' values when performing valuations over $\MMM$.

Similar analysis with respect to $\lnot \varphi$ applies to the case where $\llbracket \varphi \rrbracket_v = \FFF$.
When $\llbracket \varphi \rrbracket_v = \XXX$ for a ternary vector $v$ as above then we know that property $N_{k+1}$ holds.
Over $\MMM$, a corresponding valuation $w$ will give $\llbracket \varphi \rrbracket_w =|a_j|$, $j \leq k$, assuming the result is positive (for a negative result we refer to $\lnot \varphi$).
Then we know that $N(j+1)$ holds, 
So, if $j <k$ then the set of terms that we know that they are not implicants is larger than in the ternary case.
In addition, we know that $P(j)$ holds - with no analogous information gained in the ternary case.

Overall, we see that $\MMM$-tests are more informative than $\KKK$-tests, but, nevertheless, we do not know if it suffices to reduce the number of tests needed for complete functional verification in general (and if yes, whether such a reduction is significant).
We also do not know if we can make use of property $N(\cdot)$ to show that in special cases $\mathfrak{C}_v(\varphi) < \mathfrak{C}_f(\varphi)$ may hold.
In case property $N(\cdot)$ does not help in reducing $\mathfrak{C}_v(\varphi)$ then we have: $\mathfrak{C}_f(\varphi) \leq \mathfrak{C}_v(\varphi)$. 

Suppose that we know that $| \varphi(a_1, \ldots, a_n)| = a_i$.
Then, by Theorem~\ref{thm:main_general}, the information gained from this computation is equivalent to the one obtained by restricting ourselves to only $6$ values: $\pm a, \pm b, \pm c$, with $0 < a < b < c$, and the following mapping:
$a_j \mapsto a$, for $0 < a_j < |a_i|$ (and $a_j \mapsto -a$, for $-|a_i| < a_j < 0$); $a_j \mapsto b$ for $a_j = |a_i|$ (and $a_j \mapsto -b$ for $a_j = -|a_i|$); $a_j \mapsto c$, for $a_j > |a_i|$ (and $a_j \mapsto -c$, for $a_j < -|a_i|$).
%(when $a=0$ then $5$ values suffice). 
%However, since we do not know a-priory where the boundary between the \lq{}don\rq{}t care\rq{} and the \lq{}care\rq{} lies, it is best to make use of $n$ distinct absolute values for the arguments to be sure we gain the maximal possible information.
In fact, if we use also the value $0$ then we can be satisfied with only $5$ values, with $a = 0$, and this is the optimal number of values for maximal information gained from  the computations in case the expression we test is known to us. 

Let us look at the following simple examples.
The Boolean expressions are given as combinational circuits, where the $\lor$, $\land$ and $\lnot$ gates are interpreted as the maximum, minimum and negation respectively over $\MMM$.
For better readability, we use as before the sum, product and complement notation instead of $\lor$, $\land$ and $\lnot$.
\begin{example}
AND gate on $n$ inputs: $\varphi(x_1, \ldots, x_n) = x_1 x_2 \cdots x_n =
%In $\bb$ we need $2^n$ test vectors for complete verification, whereas in $\MMM$ only $n+1$ test vectors are needed.
\mathcal{M}_{min}(\varphi)$. Then $\mathcal{M}_{min}(\lnot \varphi) = \bar{x}_1 + \bar{x}_2 +  \cdots +  \bar{x}_n$.
Hence, $\mathfrak{C}_s(\varphi) = n+1$, and this is also the value of $\mathfrak{C}_f(\varphi)$ and of $\mathfrak{C}_v(\varphi)$.
The test vector that corresponds to $x_1 x_2 \cdots x_n$ is $(1,1, \ldots, 1)$,
and for each term $\bar{x}_i$, $i = 1, \ldots, n$, we form the test vector which assigns $x_i$ the value $-2$ and the other variables the value $1$ (it does not matter whether the value is $1$ or $-1$ as it is ``don't care'').
%When $x_i$ is assigned the value $-2$ and the other variables $1$ the result should be $-2$, and by Theorem~\ref{thm:main} the values of the other variables are ``don't care''.
%Having done $n$ such tests, for $i = 1, \ldots, n$, we assign all variables the value $1$, and check that the result is indeed $1$ for complete verification.
%By this we cover all possibilities for $\bb$.
\end{example}
%\end{exm}
%\begin{example}{\it{(Projection)}}
%$f(x_1, x_2, \ldots, x_n) = x_i$.
%To test $f$ in $\MMM$ we need only $2$ vectors, compared to $2^n$ in $\bb$:
%$x_i$ is assigned once the value $2$ and then the value $-2$, while all other variables are assigned the value $1$ at both times.
%\end{example}
%\end{exm}
%
%
%For the next example, a multiplexer, $5$ values are needed.

%By the discussion following Theorem~\ref{thm:main_general}, when the function is known to us then $5$ values always suffice, and this is the optimal number of values for efficient computations. 
\begin{example}
\label{ex:mux}
Multiplexer (MUX) with $n = 2^k$ data inputs $d_0, \ldots, d_{n-1}$ and $k$ selectors $s_0, \ldots, s_{k-1}$. It represents the function
\begin{align*}
\varphi(d_0, \ldots, d_{n-1}, & s_0, \ldots, s_{k-1}) = \mathcal{M}_{min}(\varphi) = \\
& d_0 \bar{s}_{k-1} \cdots  \bar{s}_1\bar{s}_0  + d_1 \bar{s}_{k-1} \cdots \bar{s}_1 s_0 + \cdots + d_{n-1} s_{k-1} \cdots s_1 s_0.
\end{align*}
One can show that
\begin{align*}
\mathcal{M}_{min}(\lnot \varphi) = & \\
& \bar{d}_0 \bar{s}_{k-1} \cdots  \bar{s}_1\bar{s}_0  + \bar{d}_1 \bar{s}_{k-1} \cdots  \bar{s}_1 s_0 + \cdots + \bar{d}_{n-1} s_{k-1} \cdots s_1 s_0.
\end{align*}
Here, $\mathfrak{C}_s(\varphi) = \mathfrak{C}_f(\varphi) = \mathfrak{C}_v(\varphi) = 2n$.
We can form the following $\MMM$-vectors:
for each of the $n = 2^k$ possibilities of assigning each selector variable $s_i$ the value $\infty$ or the value $-\infty$, we assign the data input $d_j$ that is selected by the corresponding assignment of values to $s_0, \ldots, s_{k-1}$ first the value $2$ and then the value $-2$, while all the other data inputs are assigned the value $1$.
In Fig.~\ref{fig:mux}, a computation of a multiplexer with $4$ data entries is shown.

When $n=2$, we get the ``If-Then-Else'' function: $\varphi(d_0, d_1, s) = d_0 \bar{s} + d_1 s = \mathcal{M}_{min}(\varphi)$.
Then, $\lnot \varphi =  (\bar{d}_0 + s) (\bar{d}_1 + \bar{s}) = \bar{d}_0 \bar{d}_1 + \bar{d}_0 \bar{s} + \bar{d}_1 s + s \bar{s}$ and 
$\mathcal{M}_{min}(\lnot \varphi) = \bar{d}_0 \bar{s} + \bar{d}_1 s$, since $\bar{d}_0 \bar{d}_1$ is redundant by the consensus rule, and $s \bar{s}$ is a contradictory term.
\begin{figure}[htb]
%\vspace{-25pt}
\centering
\scalebox{0.5}{ 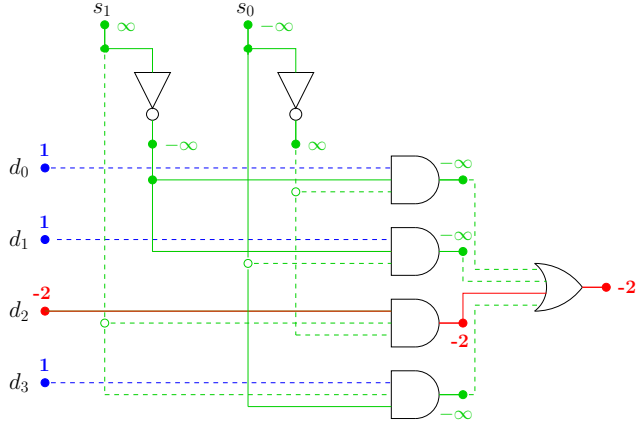 }
\caption{Multiplexer computation}
\label{fig:mux}
\vspace{-15pt}
\end{figure}
%
%This makes $2 \cdot 2^k = 2n$ different test vectors.
%
%The reason why this suffices is that all terms of the disjunction are evaluated to $-\infty$, except for the selected term, which is evaluated to the value of the selected data input: $2$ or $-2$, and this is also the value of the whole function.
%Now, the non-selected data inputs have all value $1$, which is less than $2$, and so by Theorem~\ref{thm:main} the result of the corresponding boolean computation does not depend on the truth values of these inputs and we see that all possible input vectors are covered.
\end{example}
%\end{exm}
%\vspace{9pt}

Note that the selectors in Example~\ref{ex:mux} are assigned the truth values $\pm \infty$ so that the output value is not of the selectors but of the selected data.
%Similarly, it is recommended to assign $\pm \infty$ values to variables like ``clock'' or ``enable'' when testing circuit designs.
%Here also the selector variable is given the absolute value $\infty$ or $-\infty$.
%For the same reason we assign clock or enable variables the values $\pm \infty$. 

%A major application of using $\MMM$ in equivalence verification is when we find that for a certain input vector the output of $A$ is different from that of $B$, where the difference may be in absolute value and not in sign.
%As mentioned above, we normally would run tests in which the input variables are of different absolute value.
%Then, when the design $A$ and $B$ produce different outputs, although of the same sign, indicates some variation exists, which cannot be observed in the binary setting.
%That is, with the simulations in $\MMM$ we obtain a picture which contains subtle details that are not observable when projecting into $\bb$.
%Of course, these differences in behavior between $A$ and $B$ may turn out to be ``false-negatives'', but they may indicate a real 

%%******************************************************************************************
%******************************************************************************************
%******************************************************************************************

\section{Verification of Combinational Circuits}
\label{sec:combinational}
\subsection{Comparing the Quality of Equivalent Designs}
Digital combinational circuits do not contain memory elements, hence they form Boolean expressions which represent Boolean functions.
Equivalent Boolean functions may be expressed in many ways, and a major question concerns the quality of the chosen design.
There is no definite answer to this question and it depends on the needs.
The synthesis phase of transforming a circuit design in the form of a Register Transfer Level (RTL) into a gate-level description is optimized with respect to constraints like size, timing, power consumption or ease of testability, and all these factors need to be considered when evaluating the quality of the design.
   
When a Boolean function is simple enough and is designed as an expression in DNF, as is the case of a Programmable Logic Array (PLA) and a Programmable Array Logic (PAL), then, most likely, we would prefer the minimal DNF: it is minimal in size and also its verification complexity (see Section~\ref{sec:complexity}) is the lowest.
The number of terms in a DNF of a Boolean function is, in general, exponential, and finding the minimal DNF expression, an NP-hard problem, is then double exponential in computational complexity.
Besides classical methods for minimizing the DNF, like the Quine-McCluskey algorithm \cite{Q52, M56}, which is good for small designs, other methods use heuristics for computing approximations to minimal representations, e.g. the well-known Espresso minimizer \cite{{BHMS84}}, which is also suitable for multiple outputs and multiple-level logics (and also makes use of multiple-valued logic \cite{R86} - not in the same meaning as here).

But we are not going to delve here into the intricate issues of design and optimization of circuits.
We would like to see how can we use $\MMM$-based simulations in order to find differences between two designs (or blocks in designs) which represent equivalent Boolean functions.
As we have seen, the differences in $\MMM$-tests are due to differences in implicant lengths of $\mathcal{M}(\varphi)$ and $\mathcal{M}(\lnot \varphi)$ of Boolean expressions $\varphi$, which are revealed in differences in absolute values in the $\MMM$-simulations.
These differences may represent different levels of abstractions (as is normally the case in the design flow of a circuit), but may also be interpreted as representing different degrees of truth, in the sense of fuzzy logic: a higher absolute value of a test result means a higher truth degree, or an event which is more common since it refers to a larger subset of binary input vectors that produce a similar computation.
A higher absolute value refers also to a higher ``noise stability'' (see \cite{O14}): it is less affected by a random flipping of the values of the inputs. 
%%The implication operator $a \rightarrow b$ is translated to $\lnot a \lor b$, the equivalence (XNOR) operator $a \leftrightarrow b$ is translated to $(a \land b) \lor (\lnot a \land \lnot b)$ and the exclusive-or (XOR) operator $a \oplus b$ is translated to $(a \land \lnot b) \lor (\lnot a \land b)$.
%Then, they are cha data structure of each variable is changed from boolean to integer, $\land$ is realized as the minimum function, $\lor$ as the maximum, and $\lnot$ as the negation operator of the integers. 
%\subsection{Functional Verification}
%We can run simulations in $\MMM$, analyze the results and draw conclusions on the original binary circuit.
%The simulations in $\MMM$ give us more information than the binary simulations, as is evident from Theorem~\ref{thm:main}.
%The larger range of values in $\MMM$-simulation may speed-up the search for these differences.

%Let us look at the following example.
\begin{figure}[htb]
%\vspace{-25pt}
\centering
\scalebox{0.3}{ 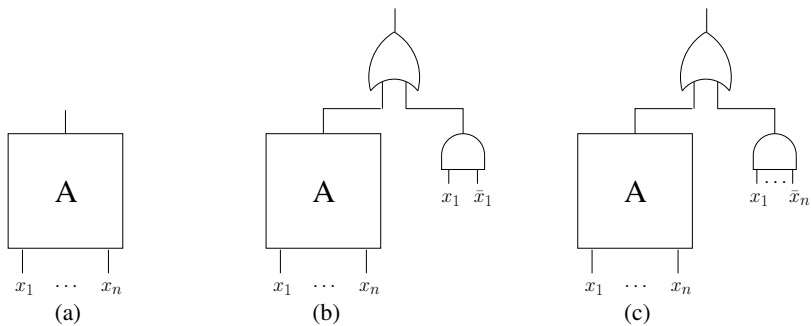 }
\caption{Quality and equivalence checking}
\label{fig:equiv}
\vspace{-15pt}
\end{figure}

\begin{example}
\label{ex:equiv}    
The circuit in Fig.~\ref{fig:equiv}(b) is identical to the circuit in Fig.~\ref{fig:equiv}(a), except for a disjunction of the output with the contradictory term $x_1 \bar{x}_1$.
Thus, the two circuits are binary equivalent and $\bb$-simulations cannot tell them apart.
Here is where $\MMM$-simulations can be of help.
Suppose we run random simulations over $\MMM$ and the input variables are assigned distinct absolute values $1, \ldots, n$ with random signs at each run.
Then, the term $x_1 \bar{x}_1$ is assigned the value $-k$ with probability $1/n$.
Hence, if there is a probability of $p_k$ for the output of the block A to be less than $-k$ when $x_1$ is assigned the value $\pm k$, then the probability of observing difference in the behavior of the two circuits (the output of the ``better'' design is of higher absolute value) in a random run is $p = \sum_{k=1}^{n-1} p_{k}$, which may be high.

Of course, when the same redundant term appears deeper in the design then it is more difficult to detect it, as it has less chance to express itself in the primary outputs.
However, a more thorough inspection into the interior of the design may reveal exceptional behavior dew to this term. \\
A similar analysis applies to a redundant conjunction with a tautology term of the form $x_1 + \bar{x}_1$.
\end{example}
\subsection{Simulations over $\MMM$}
In order to run $\MMM$-simulations on a circuit design, we need first transform it to the $\MMM$-setting.
Given a gate-level description of the design, the transformation can be executed automatically.
The  Boolean domain of values of the signal variables is replaced by $\widehat{\mathbb{Z}}$, with some chosen integer $N$, larger than all other absolute values assigned to the inputs, representing $\infty$ as the value of absolute truth.
Then, all binary operators (representing gates) are expressed through $\lor$, $\land$ and $\lnot$, and finally, these operators are defined as the maximum, minimum and negation respectively.

When our intention is to perform functional validation or a satisfiability problem then we have seen that the number of $\MMM$-tests may be significantly smaller than the number of test in the binary setting, also when we do not hope for complete verification but aim for a better coverage.

When the design is treated as a black box then, in general, the idea is to assign the input variables different absolute values in order to maximize the benefit of performing simulations over $\MMM$.
Unlike the situation in the ternary logic setting, we do not need to decide in advance which inputs are assigned ``don't care'' values.
The boundary between the ``don't care'' and ``care'' variables is dynamic and set upon {\em after} each simulation: the values that are less than the output (in absolute value) may be regarded as ``don't care''.
This means that the result of a single $\MMM$-test contains the information of both the result of the corresponding binary simulation and at the same time the result of several ternary simulations.

The larger the part of the ``don't care'' variables, the more informative is a simulation - it covers a larger set of binary input vectors.
In order to increase the size of the ``don't care'' variables, we may perform more simulations with circular shifts of the absolute values of some of the variables, without changing their signs.
This procedure is the heart of Algorithm~\ref{alg:abst}.
%
%However, since we do not know a-priory where the boundary between the ``don't care'' and the ``care'' lies, it is best to make use of $n$ distinct absolute values for the arguments to be sure we gain the maximal possible information.

%We mention again the fact that the selectors of the multiplexer in Example~\ref{ex:mux} were assigned the truth values $\pm \infty$, causing the output to be always of the selected data (in absolute value) and not of the selectors.
Similar to the assignment of truth values $\pm \infty$ to the selectors in Example~\ref{ex:mux}, it is recommended to assign $\pm \infty$ values to other control variables like ``clock'', ``enable'', ``reset'', etc, so that the value of the output will not be that of a control variable but of a data variable.

In the case of Exclusive-Or (XOR) (or its generalization to $n$ variables, the notorious Parity function) the output is always of the smaller absolute value among the inputs (see Table~\ref{table:op}): $|a \oplus b| = \min(|a|, |b|)$.
%This is clear since XOR is sensitive to any change in sign to its inputs.
This makes it more difficult to verify circuits that contain lots of XOR gates (e.g. multipliers).
%: a well-known example is the verification of multipliers in classic logic.
Here, $\MMM$ can be used in order to check whether the output is {\em larger} (in absolute value) than what is expected.

%
%In the ternary logic $\KKK$ which is commonly applied we need to choose in advance which input variables are assigned the value $\XXX$, for example in order to abstract the simulation and check whether they do not affect the output, and if they do affect - in the next test a refinement process takes place and they are assigned concrete values.
%In $\MMM$ we are in a better situation.
%We do not need to decide in advance which variables will be assigned concrete values and which don't care values.
%Simply, the smaller absolute values may be regarded as don't care as well as concrete values.
%Moreover, we have the benefit that an expression like $x \rightarrow x$ will always have positive value, possibly a ``positive don't care'', unlike the situation in $\KKK$ where it is evaluated to $\XXX$.
%This is why the value $0$ was excluded from the domain of values of $\MMM$ because $0$ behaves exactly like $\XXX$ - providing less information.
%
%The decision which values serve as don't care and which not can be done in the following way.
%Suppose the input variables are assigned values $a_1, \ldots, a_n$ which are of different absolute values.
%If the output is $\pm a_i$ then all the values that are of absolute value less than $a_i$ may be regarded as don't care - they did not affect the output.

When the design is complex then it may happen that an output variable relies on many inputs.
Hence, the number of ``don't cares'' may be small, making it less advantageous to perform simulations over $\MMM$.
In this case, if the design is not treated as a black box, we can first verify sub-blocks before verifying the whole design.
%If we perform the simulations on the complex design then it is less likely to have outputs of large absolute values (among those in the cone of influence) in real-life designs, and thus benefit from having lots of ``don't care'' values.
But this kind of behavior is not necessarily the rule.
In fact, in a design where the $\land$, $\lor$ and $\lnot$ operators are distributed randomly then by the symmetry of these operators one can expect a uniform distribution of the absolute values in the design, including among the primary outputs, when such a uniform distribution is forced upon the input values.
\subsection{An Algorithm for Obtaining a Maximal Abstract Valuation}
In what follows we use the same notation, namely $\varphi$, for a combinational circuit, the corresponding Boolean expression, and the Boolean function it represents, with arguments in $\bb$, $\KKK$ or $\MMM$, and operators of polymorphic types.

%But before, let us introduce the following definitions.
%
%
\begin{definition}
An abstraction of a vector $v \in \KKK^n$ is a vector $v' \in \KKK^n$ which is obtained from $v$ by assigning $\XXX$-values to zero or more of the binary entries of $v$.
The vector $v'$ is a strict abstraction of $v$ if $v'$ is an abstraction of $v$ and $v' \neq v$.
\end{definition}
For example, $ (\TTT, \XXX, \FFF, \XXX, \FFF)$ is a strict abstraction of $(\TTT, \XXX, \FFF, \TTT, \FFF)$.
The abstraction relation induces a partial order on $\KKK^n$.
%$v_3 = (\XXX, \XXX, \FFF, \XXX, \FFF)$, each $v_i$, $i = 2,3$, is a strict abstraction of $v_{i-1}$.
%

\begin{definition}
Given a Boolean expression $\varphi =\varphi(x_1, \ldots, x_n)$, a vector $v \in \KKK^n$ is a maximal abstract valuation with respect to $\varphi$ if $\llbracket \varphi \rrbracket_v \neq \XXX$, and for any strict abstraction $v'$ of $v$, $\llbracket \varphi \rrbracket_{v'} = \XXX$.  
\end{definition}
%
%Clearly, there may be more than one maximal abstraction of a given vector.
%
There is a one-to-one correspondence between the maximal abstract valuations $v$ satisfying $\llbracket \varphi \rrbracket_v = \TTT$ and the set of implicant terms of $\mathcal{M}(\varphi)$, and, similarly, between the maximal abstract valuations $v$ satisfying $\llbracket \varphi \rrbracket_v = \FFF$ and the implicant terms of $\mathcal{M}(\lnot \varphi)$.

\begin{definition}
A signed permutation of size $n$ is a vector $w$ which is a permutation of  $\{1, \ldots, n\}$ augmented with a sign for each number.
\end{definition}

We refer to $w$ also as a pair $(v, \sigma) \in \{-1, 1\}^n \times S_n$, and denote by $w.v$ and $w.\sigma$ the binary vector and the permutation respectively that $w$ is comprised of.
For example, $w = (3, -1, -2, 5, -4)$ is a signed permutation which is the (component-wise) product of $v = (1, -1, -1, 1, -1)$ and $\sigma = (3,1,2,5,4)$.
%There are $n! 2^n$ signed permutations of size $n$.
Given a permutation $\sigma$, we denote by $\sigma [i \lra j]$ the permutation obtained from $\sigma$ by composing it with the transposition that swaps the values $i$ and $j$.
%(not to be confused with composition of permutations as defined in cycle notation).
%(the composition is done by writing $\sigma$ in cycle notation).
For example, $(3,1,2,5,4) [2 \lra 4] = (3,1,4,5,2)$.

Algorithm~\ref{alg:abst} computes an abstraction $v'$ of a binary vector $v$ (over the set $\{ -1, 1 \}$), which is a maximal abstract valuation with respect to a combinational design $\varphi$.
As shown before, the computation of these implicant terms of $\varphi$ and $\lnot  \varphi$ plays an important role in verification of Boolean expressions.
We would like to mention that these are not necessarily prime implicants, as they reflect both the structural and the functional properties of the expression $\varphi$ and not only its functionality as do the prime implicants.

The input vector in Algorithm~\ref{alg:abst} is given as a signed permutation $w=(v, \sigma)$, and the binary vector $v$ is the projection of $w \in \MMM^n$ to $\bb^n$.
As already mentioned, when there is no knowledge on $\varphi$, then it is recommended to use different absolute values for the input vector, e.g. in the form of a signed permutation.
%The computation is done by first lifting $v$ to a signed permutation $w = (v, \sigma)$; then

The computation of a maximal abstract valuation is achieved by an iterated greedy search: if $w = w_0, w_1, \ldots, w_r = w'$ is the sequence of computed vectors then $|\llbracket \varphi \rrbracket_{w_{i-1}}| \leq |\llbracket \varphi \rrbracket_{w_i}|$, $i = 1, \ldots, r$.
The idea is the following.
When $| \llbracket \varphi \rrbracket_{w_i} | = k$ then we know that all input variables which were assigned a value $l$ with $|l| < k$ are ``don't care''.
The variable $x_{\sigma^{-1}(k)}$ is of type ``care'' (if we will map it to $X$ and perform the computation over $\KKK$ the result will be $X$). 
But there may be other variables, $x_{\sigma^{-1}(l)}$, with $| l | > k$, which are ``don't care''.
So, first we swap the absolute values (but not the signs) assigned to $x_{\sigma^{-1}(k)}$ and to $x_{\sigma^{-1}(n)}$ and perform another simulation.
Several new variables may now turn out to be ``don't care'', and we repeat the procedure of swapping, but now with $n-1$ instead of $n$ as the largest absolute value, and with the resulting $i' \geq i$ instead of $i$.
We keep iterating until the list of potential ``don't care'' variables is exhausted.  
The result is then projected to $\KKK$, providing a maximal abstract valuation which is an abstraction of $v$.

%A binary vector is represented over the set $\{ -1, 1 \}$.
%We assume that at least one input variable reaches the output of $C$.
%\begin{figure}
\begin{algorithm}[H]
\begin{algorithmic}[1]
\REQUIRE A combinational design $\varphi(x_1, \ldots, x_n)$, a signed permutation $w=(v, \sigma)$
%$binary vector $v \in \{ -1, 1 \}^n$ 
\ENSURE An abstraction $v'$ of $v$ which is a maximal abstract valuation with respect to $\varphi$
%and a corresponding signed permutation $w$
%\STATE  $w \gets (v, \sigma)$, for some $\sigma \in S_n$
%\COMMENT{$w$ is a signed permutation}
\STATE $i \gets 1; j \gets n$
\WHILE{$i < j$}
	\STATE $i \gets |\llbracket \varphi \rrbracket_{w}|$
	\STATE $w. \sigma \gets \sigma [i \lra j]$
	\STATE $j \gets j-1$
\ENDWHILE
\STATE $v' \gets p_i(w)$ \label{l:abstraction}
\COMMENT{$v'$ is the (component-wise) image of $v$ in $\KKK^n$, where if $|k| < i$ then $p_i(k) = \XXX$}
\RETURN $v'$
\end{algorithmic}
\caption{Computation of a maximal abstract valuation}
\label{alg:abst}
\end{algorithm}
%\end{figure}
%
\begin{example}
The computation shown in Fig.~\ref{fig:forest} is with input vector $(-1,3,-2,4)$: $\sigma = (1,3,2,4), v = (-1,1,-1,1)$.
The result of the main output is $2$, referring to the value assigned to the input variable $x_3$ (here $\sigma^{-1}(2) = 3$).
In order to compute a maximal abstract valuation following Algorithm~\ref{alg:abst}, we swap the values $2$ and $4$ in $\sigma$, obtaining the new input vector $(-1,3,-4,2)$.
The result of the new computation, as shown in Fig.~\ref{fig:forest1}, is $3$.
The new values of the indexes in the algorithm are $i = j =3$, and the condition of the ``while'' loop is not satisfied, so there are no more iterations.
The maximal abstract valuation vector is $(\XXX, \TTT, \FFF, \XXX)$.
\begin{figure}[hbt]
%\vspace{-25pt}
\centering
\scalebox{0.5}{ 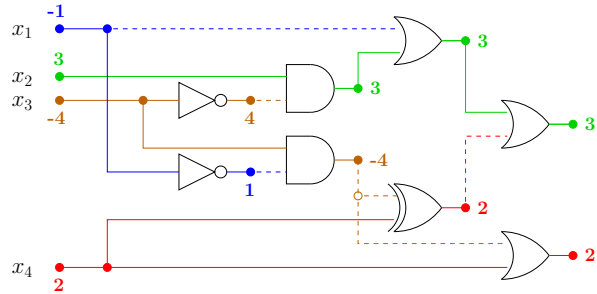 }
\caption{A new computation over $\MMM$}
\label{fig:forest1}
\vspace{-15pt}
\end{figure}
\end{example}
\begin{proposition}
Algorithm~\ref{alg:abst} computes a maximal abstract valuation of a combinational design $\varphi$.
\end{proposition}
\begin{proof}
Because at each iteration we swap absolute values which are not smaller than the absolute value of the current output then the next output cannot decrease in absolute value.
This means that the number of variables that will be mapped eventually to $\XXX$ does not decrease with each iteration.
By the end of the algorithm we get $\llbracket \varphi \rrbracket_{v'} = \llbracket \varphi \rrbracket_{p_i(w)} = p_i (\llbracket \varphi \rrbracket_{w}) = p_i (\pm i) \neq \XXX$.
%\mbox{ (or $p_i (-i)$) } \neq \XXX$.
Hence, $v'$ is an abstraction of $v$.

Each of the variables $x_{\sigma^{-1}(l)}$, with $l \geq i$ after the loop terminates, that is, a variable that is not mapped by $p_i$ to $\XXX$ (at line~\ref{l:abstraction}) had at some point a value which was of the same absolute value as the output of $\varphi$.
Hence, if at that point $x_l$ were mapped to $\XXX$ then the output over $\KKK$ of $\varphi$ would have been also $\XXX$,
let alone at the end of the algorithm where possibly more $\XXX$-s were added.
This proves that $v'$ is a maximal abstract valuation with respect to $\varphi$.  
%\qed
\end{proof}

Algorithm~\ref{alg:abst} may be incorporated in a procedure for satisfiability of a Boolean expression by a SAT solver for the purpose of pruning the search tree by leaving the binary valuation to the variables corresponding to the binary part of the resulting ternary vector of the algorithm and ignoring the ``don't care'' variables.
It may also be worth trying to flip the sign of the variable whose value corresponds to the final result of the iterations part, to see if this variable is also a ``don't care'' and then we can obtain a shorter implicant (which is not an implicant term), or the sign of the computation may then change, e.g. from $-$ to $+$ and then we found a satisfying valuation. 
Another application of Algorithm~\ref{alg:abst} is in equivalence verification, as shown in Algorithm~\ref{alg:equiv_ver}.

The number of iterations for finding a maximal abstract valuation depends on the number and lensths of the implicant terms of $\varphi$ and $\lnot \varphi$ and also on the chosen permutation (which imposes an order on the variables).  
The computation can also be computed in $\KKK$ instead of $\MMM$, but with more iterations (in average), since the boundary between the ``care'' and ``don't care'' at each iteration is not known in advance.
But then, the same line of reasoning applies to the preference of $\KKK$ over $\bb$ as the data structure for performing simulations.
%
%To summarize what we have shown, when two circuits are binary equivalent but show differences in the $\MMM$ simulations then in some cases, especially when the differences occur with positive outputs, these differences are of a qualitative characteristic.
%                        

\subsection{Equivalence Verification by Simulation}
In equivalence verification one tries to verify that two designs $A$ and $B$ are equivalent: for the same binary input vector they produce the same output.
%We do not refer here to formal equivalence verification but to simulation-based methods.
In this section we will present a procedure for equivalence checking by $\MMM$-simulations.
\begin{example}
\label{ex:equiv}    
In Fig.~\ref{fig:equiv}(a) (cal it spec) and Fig.~\ref{fig:equiv}(c) (call it imp) we see two circuits which are identical except for a disjunction of the output of imp with some conjunctive term $x_1 \cdots  \bar{x}_n$, which we assume to produce a wrong binary output. 
If there is a probability of $p_k$ for the output of spec to be less than $-k$ when performing $\MMM$-simulations of random signed permutation tests, then the probability of imp to distinguish itself from spec by producing a greater negative output value is $p = \sum_{k=1}^{n-1} p_{k} / 2^{n-k}$.
Note that this probability may be significantly greater than the probability of the two circuits to produce outputs of different signs (which happens in the rare case of the conjunctive term evaluated to the value $1$, the probability of which is $1/2^n$). 
\end{example}

In Algorithm~\ref{alg:equiv_ver} we describe a simulation procedure for checking the equivalence of two combinational circuits $A$ and $B$.
The procedure first obtains (as an output of an algorithm, could also be randomly) some binary vector $v$ and checks whether the two circuits agree on it.
If not, then a (binary) counter-example was found.
Otherwise, the procedure obtains (again, as an output of an algorithm) a corresponding signed permutation $w=(v,\sigma)$ and by Algorithm~\ref{alg:abst} two maximal abstract valuations $v_A$ and $v_B$ are returned.
If $v_A \neq v_B$ then there is a valuations in $\MMM$ on which $A$ and $B$ do not agree.
If we want to proceed manually, then we can examine the two designs on the valuation on which they do not agree and try to find the reason for that.
The partition of the graph into a spanning forest may turn out to be of great help.
If we want the procedure to be fully automatic, then we can continue with the algorithm and  try all (subject to some limit) the relevant combinations of replacing $\XXX$ values by binary ones in $v_A$ and $v_B$ and check for binary nonequivalence between the two circuits.
If no binary counter example was found then the process repeats itself with another binary vector and another signed permutation.

The idea behind the algorithm is the following.
First we compute implicant terms (but not necessarily prime implicants) for a larger coverage of the search.
Then, we look for binary nonequivalence in the environment of an $\MMM$-nonequivalent.
The latter is more common and hence can be more easily detected, see e.g. Example~\ref{ex:equiv}. 
Finally, the existence of an $\MMM$-nonequivalent hints to a possible binary nonequivalence.
%Algorithm~\ref{alg:abst} computes a maximal abstraction $v'$ of a vector $v$ with respect to a given circuit $C$.
% and a corresponding signed permutation $w$.
%The binary part of $v'$ corresponds to an implicant, not necessarily a prime implicant, of $C$ when $C(v) = \TTT$.
%Hence, we are looking for shorter implicants, and there is more chance of finding 
%However, we check also input variables 
%
\begin{algorithm}
\begin{algorithmic}[1]
\REQUIRE Two combinational designs $A, B$ on inputs $x_1, \ldots, x_n$
\ENSURE  If found -- a counter example to the equivalence of $A$ and $B$
\WHILE{true}
	\STATE Obtain a vector $v \in \{1,-1\}^n$
	\IF{$A(v) \neq B(v)$}
		\RETURN $v$
	\ENDIF
	\STATE Obtain a signed permutation $w = (v, \sigma)$ of size $n$
	\STATE $v_A \gets$ a maximal abstract valuation by Algorithm~\ref{alg:abst} on $A, w$
	\STATE $v_B \gets$ a maximal abstract valuation by Algorithm~\ref{alg:abst} on $B, w$
	\IF {$v_A \neq v_B$}
		\IF{$\exists k>0$ indexes $i$ with $v_B[i] \neq v_A[i] = \XXX$}
			\FOR{each of the $2^k$ binary combinations $u$ of flipping the values of $v[i]$} 
				\IF{$B(u) \neq B(v)$}
					\RETURN $u$
				\ENDIF
			\ENDFOR
		\ENDIF
		\STATE Repeat the process on $A$ for indexes $i$ satisfying $v_A[i] \neq v_B[i] = \XXX$		
	\ENDIF
\ENDWHILE
\end{algorithmic}
\caption{Simulation procedure for nonequivalence}
\label{alg:equiv_ver}
\end{algorithm}
%                                                                                                                         
%

%******************************************************************************************
%******************************************************************************************
%******************************************************************************************

\section{Verification of Sequential Circuits}
\label{sec:sequential}
Sequential circuits contain memory elements which introduce cycles and time dependent properties, hence they are much harder to verify.
However, at each cycle (time step), the behavior is similar to that of a combinational design, where the output as well as the memory variables are Boolean functions of the input and the memory variables.
Thus, in some common model checking methods, like bounded model checking (see e.g. \cite{BCCZ99}, \cite{SSS00}, \cite{MRS03}, \cite{M03}) the circuit is finitely unrolled and then methods like SAT-based algorithms are applied to the resulting combinational design.
Hence, the approach presented in the previous section applies also here.
Yet, $\MMM$-simulations can contribute to the verification of sequential circuits in ways which are unique to these types of circuits.
One such way is achieved by augmenting the input values with temporal data.
In what follows we hint briefly to the potential of performing $\MMM$-simulations on sequential designs. 
\subsection{Temporal Values}
One way we can benefit from using $\MMM$ instead of binary logic is by incorporating time into the variable values.
That is,an implicit global clock measures absolute time, and each new input value is assigned the time (date) of its ``birth''. 
We may use the $k$ least significant digits for the truth values (the truth part) and the other digits (the temporal part) for expressing the time of birth of that value.
%(this is equivalent to simulating with a pair of values instead of one).
At each time step the temporal parts of all the values of the input variables are incremented by $1$, while the truth parts may vary.
For example, suppose we allocate the last $3$ digits for the truth part and the other digits for the temporal part.
Then the input values may look like this (for $6$ input variables):
\begin{equation}
\begin{array}{rrrrrrrrrrrrr}
\mbox{Time 0:  }&  00\,005&    &-00\,002&    &-00\,003&    &-00\,004&    &00\,001&    &00\,006 \nonumber \\
\mbox{Time 1:  }&  -01\,004&    &-01\,005&    &01\,002&    &01\,001&    &01\,006&    &-01\,003 \nonumber \\
\mbox{Time 2:  }&  -02\,006&    &02\,003&    &02\,002&    &-02\,005&    &-02\,004&    &02\,001 \nonumber \\
\mbox{Time 3:  }&  03\,002&    &-03\,005&    &03\,001&    &-03\,003&    &-03\,006&    &03\,004 \nonumber \\
\vdots \\
\mbox{Time 40:  }&  -40\,006&    &-40\,005&    &40\,001&    &40\,002&    &-40\,004&    &40\,003 \nonumber
\end{array}
\end{equation}
Within this approach of an increasing sequence of temporal values we may still want to make sure that special control variables will obtain larger absolute values than those of the variables they interact with. 

The advantage of having temporal values is that the state of the circuit at a given time reflects directly its history: each value of a non-input variable bears its ``age'', in addition to the truth degree and input variable it originated at.
We can then observe the flow of data in space-time; e.g. pick a specific value at birth in some input variable, trace its evolution along time, until death at some time in future.
Timing considerations in the design stage may also benefit from the information within temporal values.  
\subsection{Initialization.}
In the setting of ternary logic, one starts from an ``all-$\XXX$'' state and simulates with a sequence of binary input vectors until reaching a complete binary state, thus finding a ``universal'' initialization sequence. 
When performing any simulation task over $\MMM$ with ``time stamp'' as above then at the same time we are also conducting an initialization test at the background.
Moreover, at each time step $k$ a new initialization test starts.
Thus, if we are interested in the shortest initialization sequence, we can check at each time step $l$ the lowest temporal part $k$ that exists in the values of the variables of that state, which refers to an initialization sequence of length $(l-k)+2$.
Since the input values are incremented in absolute values at each time step, then, by Theorem~\ref{thm:main_general}, when reaching a state in which all temporal values smaller than $k$ already vanished then this is equivalent to the disappearing of the $\XXX$ values in the ternary initialization.
%the domain of values is pairs $(t, a) \in \bbbn \times \MMM$, where $t$ represents time and $a$ represents truth value as before.
%Then the order is lexicographic: $(t_1, a_1) < (t_2, a_2)$ if and only if $t_1 < t_2$ or $t_1 = t_2$ and $a_1 < a_2$.
%
%\subsection{Stuck at Values.}
%After applying an increasing sequence of temporal values to the inputs we may at some time step want to change direction and start decreasing the temporal values.
%If a group of variables retain their large temporal values then we know these values are invariant to any future inputs.
%By the way, a group of such variables may exhibit a periodic behavior and need not be stuck at the same values, but the above method will detect this, probably faulty, behavior.
%
\subsection{Prioritizing.}
To a certain extent, it is possible to manipulate the flow of data in the design.
For example, the absolute value of the output of a XOR or XNOR gate equals the minimum of the absolute values of the inputs.
Then, a prioritizing methodology may be applied to drive desired inputs toward the outputs by assigning them smaller absolute values so that they will propagate through these gates in a design full of them.
Similar methods may be applied in order to increase the coverage of elements like signals, gates or latches in simulations by forcing the data to pass through these elements.
Formal or semi-formal methods may also be applied here.
Otherwise, we can measure the coverage performance of a simulation sequence in terms of the coverage of the graph representation by the trees that correspond to the values at the primary outputs.
\subsection{Composition of Blocks.}
When a design is composed of several blocks then we may run $\MMM$-simulations in a way that reflects this higher order partition.
For example, when there is little overlap between the inputs of the blocks then the input values may be grouped by absolute values according to the blocks, possibly assigning higher absolute values to blocks that are of shorter distance to the primary outputs.
In this way, we shift attention to the hierarchical structure of the design and to the interactions and dependencies between the blocks rather than to the more detailed structure inside the blocks.
\subsection{Equivalence Verification.}
The discussion and methods presented when considering combinational designs can be extended to sequential ones. 
%Circuits that are supposed to be binary equivalent may produce different $\MMM$ values.
%If the different values are also of different sign then the circuits are not binary equivalent.
%If the signs are the same but the absolute values differ then it may indicate a potential binary nonequivalence or the existence of a potential redundant part.
%It may also be the case that none of the above is the case.
%By Proposition~\ref{prop:path}, we can trace the output values all the way towards the inputs (mostly in one of the previous time steps) and try to analyze why are the values different.
%
As for comparing the qualities of the designs, we refer to \cite{CHR12} for a somewhat related work.
%In some cases, as we have seen in Subsection~\ref{subsec:eqver} higher $\MMM$ absolute values refer to a better design.
%
\subsection{Generating Assertions.}
When trying to formally verify sequential circuits, whether for property or for equivalence checking, it is almost unavoidable but to try and break the problem into sub-problems to be verified first.
This incremental methodology requires the generation of potential assertions, also referred to as lemmas, and the more refined MVL may be of help here.
In equivalence verification we can find correlations between variables, applying probabilistic methods if needed, in a more accurate manner over $\MMM$ since the spread of values is wider.
% and also since the values refer to the input variables of their origin.
The designer may also provide refined assertions over $\MMM$ for assertion-based verification and simulation.
For example, if the designer knows that some property should hold under an assumption that relies on specific input values then the property may be checked with these input values being of higher absolute value than other input values, to make sure that the output does not depend in this case on other inputs.
Assertions may also refer to the temporal values of the variables, conducting an explicit model checking over $\MMM$.
For example, properties may include exact absolute time and exact delays by referring to the temporal part of the clock variable, so that it becomes explicit and natural to express properties of Metric Temporal Logic (MTL) \cite{K90}, \cite{HOW13} over $\mathbb{Z}$. 
These ideas need to be further explored.
%
%

%******************************************************************************************
%******************************************************************************************
%******************************************************************************************

\section{Conclusion}
\label{sec:conclusion}
Simulations over the multiple-valued logic $\MMM$ are more refined and informative than over binary and ternary logics, thus providing a novel potential approach to the complex task of verification of HW designs.
A state of the system is enriched with data that includes degrees of truth and, for sequential designs, identity stamps like ``place'' and ``date of birth''.
We presented the theory behind computations and verification over $\MMM$, and discussed general directions, including algorithms, for applying $\MMM$-simulations to different verification tasks.
Future goals include implementing and checking these ideas on real HW designs and developing specific and elaborate strategies and algorithms.  

\end{document}